\documentclass[12pt]{article}
\usepackage{amsmath, amssymb, stmaryrd, setspace, nicefrac, multicol, amsthm, geometry,endnotes}
\usepackage{fullpage}
\usepackage[utf8]{inputenc}
\usepackage{mathrsfs}
\usepackage{stmaryrd}

\newtheorem{thm}{Theorem}

\newtheorem{lem}[thm]{Lemma}
\newtheorem{cor}[thm]{Corollary}

\theoremstyle{definition}
\newtheorem*{defn}{Definition}

\newtheorem{obs}[thm]{Observation}

\numberwithin{thm}{section}

\newcommand{\harp}{\upharpoonright}
\newcommand{\ent}{\textrm{H}}
\newcommand{\mi}{\textrm{I}}

\usepackage{xcolor}	
	\usepackage{soul}

\title{Finite-State Mutual Dimension}
\author{Adam Case\footnote{Part of this research was carried out while the authors participated in the program ``Equidistribution: Arithmetic, Computational and Probabilistic Aspects'' at the National University of Singapore Institute for Mathematical Sciences in 2019.}\\ \small Drake University
        \and 
        Jack H. Lutz\footnotemark[1] \footnote{This research was supported in part by National Science Foundation grants 1545028 and 1900716.}\\ \small Iowa State University}
\date{} 
\begin{document}
\maketitle
\abstract{In 2004, Dai, Lathrop, Lutz, and Mayordomo defined and investigated the \emph{finite-state dimension} (a finite-state version of algorithmic dimension) of a sequence $S \in \Sigma^\infty$ and, in 2018, Case and Lutz defined and investigated the \emph{mutual (algorithmic) dimension} between two sequences $S \in \Sigma^\infty$ and $T \in \Sigma^\infty$. In this paper, we propose a definition for the \emph{lower} and \emph{upper finite-state mutual dimensions} $mdim_{FS}(S:T)$ and $Mdim_{FS}(S:T)$ between two sequences $S \in \Sigma^\infty$ and $T \in \Sigma^\infty$ over an alphabet $\Sigma$. Intuitively, the finite-state dimension of a sequence $S \in \Sigma^\infty$ represents the density of finite-state information contained within $S$, while the finite-state mutual dimension between two sequences $S \in \Sigma^\infty$ and $T \in \Sigma^\infty$ represents the density of finite-state information shared by $S$ and $T$. Thus ``finite-state mutual dimension'' can be viewed as a ``finite-state'' version of mutual dimension and as a ``mutual'' version of finite-state dimension.

The main results of this investigation are as follows. First, we show that finite-state mutual dimension, defined using \emph{information-lossless finite-state compressors}, has all of the properties expected of a measure of \emph{mutual information}. Next, we prove that finite-state mutual dimension may be characterized in terms of \emph{block mutual information rates}. Finally, we provide necessary and sufficient conditions for two normal sequences to achieve $mdim_{FS}(S:T) = Mdim_{FS}(S:T) = 0$.}

\section{Introduction}

The study of \emph{algorithmic dimension} has yielded various mechanisms for quantifying the density of information contained within infinite objects, such as points in Euclidean space \cite{jLutMay08} and sequences \cite{jLutz03a}. Recent investigations into the dimensions of points and sequences have produced new characterizations of classical Hausdorff dimension \cite{jHitc05,jLutLut18,jLutLut21} and insights into self-similar fractal geometry \cite{jLutMay08,jDLMT13,jXLMP14}, among other results. Originally defined in terms of \emph{gales} (a generalization of \emph{martingales}) \cite{jLutz03a}, the \emph{dimension} $dim(S)$ and \emph{strong dimension} $Dim(S)$ of a sequence $S \in \Sigma^\infty$ were shown to have the characterizations
\[
dim(S) = \displaystyle\liminf_{n \rightarrow \infty}\frac{K(S \harp n)}{n \log |\Sigma|}
\]
and
\[
Dim(S) = \displaystyle\limsup_{n \rightarrow \infty}\frac{K(S \harp n)}{n\log |\Sigma|},
\]
where $K(S \harp n)$ is the \emph{Kolmogorov complexity} of the first $n$ symbols of $S$ \cite{jMayo02,jAHLM07}. These characterizations show that $dim(S)$ and $Dim(S)$ can be thought of as the lower and upper densities of \emph{algorithmic information} contained within $S$. The algorithmic dimension and \emph{algorithmic randomness} of sequences have been shown to have interesting relationships. For example, if a sequence $S \in \Sigma^\infty$ is \emph{(algorithmically) random}, then $dim(S) = 1$. However, not all sequences that achieve $dim(S) = 1$ are necessarily random \cite{jLutz03a}.

The notion of the dimension of a sequence has been adapted to operate within different contexts in the fields of computability and information theory. For example, Dai, Lathrop, Lutz, and Mayordomo developed the notion of \emph{finite-state dimension}, which is a finite-state version of algorithmic dimension \cite{jDaLaLuMa04}. In their paper, the authors define finite-state dimension in terms of \emph{finite-state gamblers}. In \cite{jDaLaLuMa04} and \cite{jAHLM07} the authors show that the \emph{finite-state dimension} $dim_{FS}(S)$ and \emph{finite-state strong dimension} $Dim_{FS}(S)$ of a sequence $S \in \Sigma^\infty$ may be characterized by
\begin{align}\label{compratio}
dim_{FS}(S) = \inf\bigg\{\displaystyle\liminf_{n \rightarrow \infty}\frac{|C(S \harp n)|}{n \log \Sigma} \, \bigg | \, C \text{ is an ILFSC} \bigg\}
\end{align}
and
\begin{align}\label{scompratio}
Dim_{FS}(S) = \inf\bigg\{\displaystyle\limsup_{n \rightarrow \infty}\frac{|C(S \harp n)|}{n \log \Sigma} \, \bigg | \, C \text{ is an ILFSC} \bigg\},
\end{align}
where $C$ is an \emph{information-lossless finite-state compressor} (ILFSC) and $|C(S \harp n)|$ is the length of the output that $C$ produces when given the first $n$ symbols of $S$ as input. These quantities can be thought of as the lower and upper densities of \emph{finite-state information} contained within $S$ and are also known as the \emph{lower} and \emph{upper compression ratios} of $S$ as studied by Ziv and Lempel \cite{jZivLem78}.

Other characterizations of finite-state dimension have been shown. For example, Bourke, Hitchcock, and Vinodchandran proved that the lower and upper finite-state dimensions of a sequence $S \in \Sigma^\infty$ are equal to the lower and upper \emph{block entropy rates} of $S$, respectively (i.e., the lower and upper limiting normalized entropies of the frequencies of aligned blocks of symbols contained within $S$) \cite{jBoHiVi05}. In a recent paper, Kozachinskiy and Shen show that finite-state dimension can also be characterized in terms of the entropy rates of non-aligned blocks of symbols and in terms of superadditive calibrated functions on strings \cite{cKozShe2019}.

There have been several interesting explorations into the relationships between finite-state dimension and the concept of \emph{normality}, which was introduced by Borel in 1909 \cite{jBor09}. A sequence $S \in \Sigma^\infty$ is \emph{normal} if every string of the same length occurs with the same limiting frequency within $S$. Normality can be viewed as a weaker form of randomness, since every algorithmically random sequence is also normal but not vice-versa. In fact, it has been shown that a sequence is normal if and only if $dim_{FS}(S) = 1$ \cite{jDaLaLuMa04,jBoHiVi05}. Thus the normal sequences can be completely characterized as the sequences that achieve finite-state dimension one. This equivalence has recently been quantitatively refined using the Kullback-Leibler divergence \cite{jXiLuMaSt21}.

Another way in which the dimensions of sequences has been adapted to fit other contexts within information theory can be found in the development of \emph{mutual dimension}, which was introduced in 2015 by the present authors in \cite{jCasLut15}. In this paper, the authors defined the mutual dimension between two points in Euclidean space and showed that it has all the properties expected of a measure of mutual information, including several \emph{data processing inequalities}. In 2018, the same authors extended this framework to sequences and defined the \emph{lower} and \emph{upper mutual dimensions}, $mdim(S:T)$ and $Mdim(S:T)$, respectively, between two sequences $S \in \Sigma^\infty$ and $T \in \Sigma^\infty$ by
\[
mdim(S) = \displaystyle\liminf_{n \rightarrow \infty}\frac{I(S \harp n:T \harp n)}{n \log |\Sigma|}
\]
and
\[
Mdim(S) = \displaystyle\limsup_{n \rightarrow \infty}\frac{I(S \harp n:T \harp n)}{n\log |\Sigma|},
\]
where $I(S \harp n: T \harp n)$ is the \emph{algorithmic mutual information} between the first $n$ bits of $S$ and $T$ \cite{jCasLut18}. The algorithmic mutual information $I(u:w)$ between two strings $u \in \Sigma^*$ and $w \in \Sigma^*$ is
\[
I(u:w) = K(w) - K(w\,|\,u),
\]
where $K(w\,|\,u)$ is the \emph{Kolmogorov complexity} of $w$ \emph{given} $u$. However, this quantity can also be characterized by
\[
I(u:w) = K(u) + K(w) - K(u,w) + o(|u|),
\]
where $K(u,w)$ is the \emph{joint Kolmogorov complexity} of $u$ and $w$. (The interested reader may refer to \cite{bLiVit08} for an in-depth discussion on algorithmic mutual information.) Therefore, we can view the lower and upper mutual dimensions as the lower and upper densities of algorithmic mutual information shared by two sequences. In the same paper, the authors demonstrate that, if two sequences $S \in \Sigma^\infty$ and $T \in \Sigma^\infty$ are \emph{independently random}, then $Mdim(S:T) = 0$. However, they also show that not all pairs of sequences that achieve mutual dimension zero are necessarily independently random \cite{jCasLut18}.

The purpose of this article is to develop a notion of \emph{finite-state mutual dimension}, which includes defining it using information-lossless finite-state compressors, proving that it can be characterized in terms of block entropy rates, and exploring its relationship with normal sequences. The outline of this article is as follows. In Section 2, we define the \emph{joint} compression ratio of two strings as well as the \emph{mutual compression ratio} between two strings. Using Ziv and Lempel's Generalized Kraft Inequality \cite{jZivLem78}, we establish several relationships between the \emph{Shannon entropy} of the \emph{joint block frequency} of two strings $u \in \Sigma^n$ and $w \in \Sigma^n$ and the joint compression ratio of $u$ and $w$. Using these relationships, we are able to prove the basic properties of the mutual compression ratio between finite strings. In Section 3, we extend the notion of the mutual compression ratio to infinite sequences and use it to define the lower and upper finite-state mutual dimensions. We prove an important theorem regarding the interchangeability of the iterated limits within the definition of finite-state mutual dimension, which we then use to prove the basic properties of finite-state mutual dimension. In Section 4, we introduce the \emph{lower} and \emph{upper block mutual information rates} between two sequences $S \in \Sigma^\infty$ and $T \in \Sigma^\infty$ and show that they are equal to the lower and upper finite-state mutual dimensions, respectively. In Section 5, we obtain a result regarding the independence of sequences at the finite-state level. Specifically, we prove that, if $R_1 \in \Sigma^\infty$ and $R_2 \in \Sigma^\infty$ are normal, then the sequence $(R_1,R_2) \in (\Sigma \times \Sigma)^\infty$ is normal if and only if $Mdim_{FS}(R_1:R_2) = 0$, where $(R_1,R_2)$ is the sequence obtained by pairing the symbols of $R_1$ and $R_2$ at the same index.


\section{Joint and Mutual Compression Ratios of Strings}
In this section, we define and investigate the \emph{joint compression ratio} of two strings. To do this, we make use of some relationships between the compression ratios and entropies of the relative frequencies of strings that were originally established by Ziv and Lempel \cite{jZivLem78} and further examined by Sheinwald \cite{jShei94}. We also introduce the \emph{mutual compression ratio} between two strings and explore its properties.

In this paper, we assume that $\Sigma$ is an alphabet consisting of $k$ symbols. We write $\Sigma^*$ to represent the set of all strings over $\Sigma$ and $\Sigma^\infty$ to represent the set of all sequences over $\Sigma$. The \emph{length} of a string $u \in \Sigma^*$ is denoted by $|u|$ and we represent the set of all strings of length $n \in \mathbb{N}$ by $\Sigma^n$. For any sequence $S \in \Sigma^\infty$, we write $S \harp n$ for the first $n \in \mathbb{N}$ symbols of $S$. For any string $u \in \Sigma^*$ and sequence $S \in \Sigma^\infty$, we write $u[i]$ and $S[i]$ for the $i^{th}$ bit of $u$ and the $i^{th}$ bit of $S$, respectively. For any two strings $u \in \Sigma^n$ and $w \in \Sigma^n$, we write $(u,w)$ to represent the string
\[
(u,w) = (u[1],w[1])(u[2],w[2])\cdots(u[n],w[n]) \in (\Sigma \times \Sigma)^n.
\]
Note that the lengths of $u$ and $w$ must be equal in order to use the notation $(u,w)$ for strings. Similarly, for any two sequences $S \in \Sigma^\infty$ and $T \in \Sigma^\infty$, we write $(S,T)$ to represent the sequence
\[
(S,T) = (S[1],T[1])(S[2],T[2])\cdots \in (\Sigma \times \Sigma)^\infty.
\]
We will write $\log$ for the base-2 logarithm function and $\log_k$ for the base-k logarithm function.

A \emph{discrete probability measure} $\alpha$ on a finite set $\mathcal{X}$ is a function $\alpha: \mathcal{X} \rightarrow [0,1]$ such that
\[
\displaystyle\sum_{x \in \mathcal{X}}\alpha(x) = 1.
\]
\begin{defn}
Let $\alpha$ be a discrete probability measure on $\mathcal{X}$. The \emph{Shannon entropy} of $\alpha$ is
\[
\ent(\alpha) = \displaystyle\sum_{x \in \mathcal{X}}\alpha(x)\log\frac{1}{\alpha(x)}.
\]
\end{defn}

If $\alpha$ is a discrete probability measure on $\mathcal{X} \times \mathcal{X}$ (sometimes called a \emph{joint probability measure} on $\mathcal{X}$), we will write $\alpha(x,y)$ to denote the value $\alpha((x,y))$ assigned to the pair $(x,y)$ by $\alpha$. The \emph{first} and \emph{second marginal probability measures} of $\alpha$ are the probability measures $\alpha_1$ and $\alpha_2$ on $\mathcal{X}$ defined by
\[
\alpha_1(a) = \displaystyle\sum_{b \in \mathcal{X}}\alpha(a,b) \text{\,\,\,\, and \,\,\,\,} \alpha_2(b) = \displaystyle\sum_{a \in \mathcal{X}}\alpha(a,b),
\]
respectively.

The following theorem states well-known inequalities regarding Shannon entropy \cite{bCovTho06}.
\begin{thm}\label{Sha}
Let $\alpha$ be a probability measure on $\mathcal{X} \times \mathcal{X}$.
\begin{enumerate}
\item $\emph{\ent}(\alpha) \leq \emph{\ent}(\alpha_1) + \emph{\ent}(\alpha_2)$.
\item $\max\{\emph{\ent}(\alpha_1), \emph{\ent}(\alpha_2)\} \leq \emph{\ent}(\alpha)$.
\item If $\sum_{a \in \Sigma}\alpha(a,a) = 1$, then $\emph{\ent}(\alpha) = \emph{\ent}(\alpha_1) = \emph{\ent}(\alpha_2)$.
\end{enumerate}
\end{thm}
 
For any $n,\ell \in \mathbb{Z}^+$ such that $n$ is a multiple of $\ell$ and all $x \in \Sigma^\ell$ and $u \in \Sigma^n$, we denote the number of \emph{block occurrences} of $x$ in $u$ to be
\[
\#_\Box(x,u) = \bigg|\bigg\{m \leq \frac{n}{\ell} \,\Big|\, u[m\ell\ldots (m+1)\ell - 1] = x\bigg\}\bigg|
\]
and the \emph{block frequency} of $x$ in $u$ by the function $\pi_u: \Sigma^* \rightarrow \mathbb{Q}_{[0,1]}$ defined by
\[
\pi_u(x) = \frac{l}{n}\#_\Box(x,u),
\]
where $\mathbb{Q}_{[0,1]}$ is the set of all rationals in $[0,1]$. For all $n,\ell \in \mathbb{Z}^+$ such that $n$ is a multiple of $\ell$ and $u \in \Sigma^n$, we denote the restriction of $\pi_u$ to strings in $\Sigma^{\ell}$ by $\pi_u^{(\ell)}$. It is important to note that $\pi_u^{(\ell)}$ represents a discrete probability measure on the finite set $\Sigma^{\ell}$.

For all $x,y \in \Sigma^\ell$ and $u,w \in \Sigma^n$, we denote the \emph{joint block frequency} of $x$ in $u$ and $y$ in $w$ by the function $\pi_{u,w}: \Sigma^* \times \Sigma^* \rightarrow \mathbb{Q}_{[0,1]}$ defined by
\[
\pi_{u,w}(x,y) = \frac{\ell}{n}\#_\Box((x,y),(u,w)).
\]
We denote the restriction of $\pi_{u,w}$ to the pairs of strings in $\Sigma^{\ell} \times \Sigma^{\ell}$ by $\pi_{u,w}^{(\ell)}$. Once again, we note that $\pi_{u,w}^{(\ell)}$ is a discrete probability measure on $\Sigma^\ell \times \Sigma^\ell$. It is easy see that, for all $x,y \in \Sigma^\ell$, $\pi_{(u,w)}^{(\ell)}((x,y)) = \pi_{u,w}^{(\ell)}(x,y)$. Also, it is also important to observe that the first and second marginal probability measures of $\pi_{u,w}^{(\ell)}$ are $\pi_u^{(\ell)}$ and $\pi_w^{(\ell)}$, respectively.

By applying Theorem \ref{Sha} to $\pi_{u,w}^{(\ell)}$, where $\alpha = \pi_{u,w}^{(\ell)}$, $\alpha_1 = \pi_u^{(\ell)}$, and $\alpha_2 = \pi_w^{(\ell)}$, we obtain the following corollary.

\begin{cor}\label{Shac}
For every $\ell \in \mathbb{Z}^+$ and every $n \in \mathbb{N}$ and $u,w \in \Sigma^n$ such that $n$ is a multiple of $\ell$, the following hold.
\begin{enumerate}
\item $\emph{\ent}(\pi_{u,w}^{(\ell)}) \leq \emph{\ent}(\pi_u^{(\ell)}) + \emph{\ent}(\pi_w^{(\ell)})$.
\item $\max\{\emph{\ent}(\pi_u^{(\ell)}), \emph{\ent}(\pi_w^{(\ell)})\} \leq \emph{\ent}(\pi_{u,w}^{(\ell)})$.
\item $\emph{\ent}(\pi_{u,u}^{(\ell)}) = \emph{\ent}(\pi_u^{(\ell)})$.
\item $\emph{\ent}(\pi_{u,w}^{(\ell)}) = \emph{\ent}(\pi_{w,u}^{(\ell)})$.
\end{enumerate}
\end{cor}

We now proceed to discuss the finite-state compressibility of strings. A \emph{finite-state compressor} (FSC) $C$ on $\Sigma$ is a 4-tuple
\[
C = (Q, \delta, \nu, q_0),
\]
where $Q$ is a nonempty finite set of \emph{states}, $\delta: Q \times \Sigma \rightarrow Q$ is the \emph{transition function}, $\nu: Q \times \Sigma \rightarrow \{0,1\}^*$ is the \emph{output function}, and $q_0$ is the \emph{initial state}. We define the \emph{extended transition function} $\delta^*: Q \times \Sigma^* \rightarrow Q$ by the recursion
\begin{align*}
\delta^*(q,\lambda) &= q,\\
\delta^*(q,wa) &= \delta(\delta^*(q,w),a),
\end{align*}
for all $q \in Q$, $u \in \Sigma^*$, and $a \in \Sigma$. The output function $\nu$ is defined by the recursion
\begin{align*}
\nu(q_0,\lambda) &= \lambda,\\
\nu(q_0,wa) &= \nu(q_0,w)\nu(\delta^*(q_0,w), a),
\end{align*}
for all $u \in \Sigma^*$ and $a \in \Sigma$. The \emph{output} of $C$ on the input string $u \in \Sigma^*$ is denoted by $C(u) = \nu(q_0,u)$. An \emph{information-lossless finite-state compressor} (ILFSC) $C$ is an FSC where the function $f: \Sigma^* \rightarrow \{0,1\}^* \times Q$, defined by $f(u) = (C(u),\delta^*(u))$, is one-to-one.

The \emph{compression ratio of} $u \in \Sigma^n$ attained by an ILFSC $C$ on $\Sigma$ is
\[
\rho_C(u) = \frac{|C(u)|}{n\log k}.
\]

Likewise, the \emph{joint compression ratio of} $u \in \Sigma^n$ and $w \in \Sigma^n$ attained by an ILFSC $C$ on $\Sigma \times \Sigma$ is
\[
\rho_C(u,w) = \frac{|C((u,w))|}{n\log k}.
\]

\begin{defn}
The $r$-\emph{state compression ratio of} $u \in \Sigma^n$ is
\[
\rho_r(u) = \min\big\{\rho_C(u) \,\bigg|\, C \text{ is an ILFSC on } \Sigma \text{ that has } r \text{ states} \big\}.
\]
\end{defn}

\begin{defn}
The $r$-\emph{state joint compression ratio of} $u \in \Sigma^n$ and $w \in \Sigma^n$ is
\[
\rho_r(u,w) = \min\big\{\rho_C(u,w) \,\bigg|\, C \text{ is an ILFSC on } \Sigma \times \Sigma \text{ that has } r \text{ states} \big\}.
\]
\end{defn}
It is important to note that $\rho_r((u,w))$ is the $r$-state compression ratio of the string $(u,w) \in (\Sigma \times \Sigma)^n$ and $\rho_r(u,w)$ is the $r$-state joint compression ratio of $u \in \Sigma^n$ and $w \in \Sigma^n$.

The following lemma was proven by Ziv and Lempel in \cite{jZivLem78}.

\begin{lem}[Generalized Kraft Inequality \cite{jZivLem78}]\label{kraft}
For any ILFSC $C$ on $\Sigma$ with a state set $Q = \{q_1, q_2, \ldots, q_s\}$,
\[
\displaystyle\sum_{w \in \Sigma^{r}}2^{-L_C(w)} \leq s^2 \bigg (1 + \log{\frac{s^2 + k^{r}}{s^2}} \bigg ),
\]
where
\[
L_C(w) = \displaystyle\min_{q \in Q}\{|C_q(w)|\}
\]
and $C_q$ is the ILFSC that is like $C$ except that it uses $q$ as the start state.
\end{lem}

For the remainder of this article, we will make use of the following family of functions. For each $s,k,r \in \mathbb{Z}^+$, let
\[
f_s^k(r) = \frac{\log\bigg(s^2 \big (1 + \log{\frac{s^2 + k^{r}}{s^2}} \big )\bigg)}{r\log k}.
\]
It is easy to see that, for any fixed $s,k \in \mathbb{Z}^+$,
\[
\displaystyle\lim_{r \rightarrow \infty}f_s^k(r) = 0.
\]
For the sake of reducing notation, we write $u_r$ for the prefix
\[
u \harp \Big \lfloor \frac{|u|}{r} \Big \rfloor \cdot r,
\]
where $r \in \mathbb{Z}^+$, $u \in \Sigma^*$ such that $|u| \geq r$. We also write $S^n_r$ for the prefix
\[
S \harp \Big \lfloor \frac{n}{r} \Big \rfloor \cdot r,
\]
where $S \in \Sigma^{\infty}$.

\begin{obs}\label{obs1}
Let $C_1$ be an ILFSC on $\Sigma$ and $C_2$ be an ILFSC on $\Sigma \times \Sigma$. For every $r,n \in \mathbb{Z}^+$ and every $u \in \Sigma^n$ and $w \in \Sigma^n$ such that $r \leq n$,
\[
\rho_{C_1}(u_r) \leq \rho_{C_1}(u) + \Big \lfloor \frac{n}{r}  \Big \rfloor^{-1}.
\]
and
\[
\rho_{C_2}(u_r, w_r) \leq \rho_{C_2}(u,w) + \Big \lfloor \frac{n}{r}  \Big \rfloor^{-1},
\]
where $u_r = u \harp \Big \lfloor \frac{n}{r} \Big \rfloor \cdot r$ and $w_r = w \harp \Big \lfloor \frac{n}{r} \Big \rfloor \cdot r$.
\end{obs}

\begin{proof}
We proceed to prove the first inequality. The following inequality holds for all $a,b,x \in \mathbb{R}^+$ such that $a \leq b$ and $x \leq b$,
\[
\frac{a - x}{b - x} \leq \frac{a}{b}.
\]
Using this inequality, observe that
\begingroup
\allowdisplaybreaks
\begin{align*}
\rho_{C_1}(u) &= \frac{|C_1(u)|}{n \log k}\\
					&\geq \frac{|C_1(u \harp \Big \lfloor \frac{n}{r}  \Big \rfloor \cdot r)|}{n \log k}\\
					&\geq \frac{|C_1(u \harp \Big \lfloor \frac{n}{r}  \Big \rfloor \cdot r)| - (n \bmod r)\log k}{n \log k - (n \bmod r)\log k}\\
					&= \frac{|C_1(u \harp \Big \lfloor \frac{n}{r}  \Big \rfloor \cdot r)| - (n \bmod r)\log k}{\Big \lfloor \frac{n}{r}  \Big \rfloor \cdot r  \log k}\\
					&\geq \frac{|C_1(u \harp \Big \lfloor \frac{n}{r}  \Big \rfloor \cdot r)|}{\Big \lfloor \frac{n}{r}  \Big \rfloor \cdot r  \log k} -  \frac{1}{\Big \lfloor \frac{n}{r}  \Big \rfloor}\\
					&= \frac{|C_1(u \harp \Big \lfloor \frac{n}{r}  \Big \rfloor \cdot r)|}{\Big \lfloor \frac{n}{r}  \Big \rfloor \cdot r  \log k} - \Big \lfloor \frac{n}{r}  \Big \rfloor^{-1}\\
					&= \rho_{C_1}(u \harp \Big \lfloor \frac{n}{r} \Big \rfloor \cdot r) - \Big \lfloor \frac{n}{r}  \Big \rfloor^{-1}\\
					&= \rho_{C_1}(u_r) - \Big \lfloor \frac{n}{r}  \Big \rfloor^{-1}
\end{align*}
\endgroup
The proof of the second inequality is identical to the proof of the first inequality. \qedhere
\end{proof}

The following lemma describes an inequality that was noted by Sheinwald in \cite{jShei94}. Originally, Ziv and Lempel noted a similar inequality in \cite{jZivLem78}.
\begin{lem}[Sheinwald \cite{jShei94}]\label{low}
Let $C$ be an ILFSC on $\Sigma$. For every $\ell,n \in \mathbb{Z}^+$ and $u \in \Sigma^n$ such that $n$ is a multiple of $\ell$,
\[
\rho_C(u) \geq \frac{1}{\ell \log k}\displaystyle\sum_{x \in \Sigma^{\ell}}\pi_u^{(\ell)}(x)L_C(x).
\]
\end{lem}
It is worth noting that Ziv and Lempel and Sheinwald originally used the notation $P(x,u)$ in place of $\pi_u^{(\ell)}(x)$.

\begin{lem}\label{hc}
Let $C$ be an ILFSC on $\Sigma$ with $s \in \mathbb{Z}^+$ states. For every $\ell,n \in \mathbb{N}$ and $u \in \Sigma^n$ such that $\ell \leq n$,
\[
\frac{\emph{\ent}(\pi_{u_\ell}^{(\ell)})}{\ell \log k} - \rho_C(u) \leq \Big \lfloor \frac{n}{\ell}  \Big \rfloor^{-1} + f_s^k(\ell),
\]
where $u_\ell = u \harp \Big \lfloor \frac{n}{\ell} \Big \rfloor \cdot \ell$ and $\lim_{m \rightarrow \infty}f_s^k(m) = 0$.
\end{lem}

\begin{proof}
The following proof uses similar reasoning as Sheinwald's proof that the upper compression ratio of a sequence $S \in \Sigma^{\infty}$ is equal to the upper block entropy rate of $S$ \cite{jShei94}. Let $C$ be an ILFSC on $\Sigma$ with $s \in \mathbb{Z}^+$ states. By the first inequality stated in Observation \ref{obs1} and Lemma \ref{low},
\begin{align}\label{firstineq}
\frac{\ent(\pi_{u_\ell}^{(\ell)})}{\ell \log k} - \rho_C(u) &\leq \frac{1}{\ell \log k}\displaystyle\sum_{x \in \Sigma^{\ell}}\pi_{u_\ell}^{(\ell)}(x)\log{\frac{1}{\pi_{u_\ell}^{(\ell)}(x)}} - \frac{1}{\ell \log k}\displaystyle\sum_{x \in \Sigma^{\ell}}\pi_{u_\ell}^{(\ell)}(x)L_C(x) + \Big \lfloor \frac{n}{\ell}  \Big \rfloor^{-1} \nonumber\\
&= \frac{1}{\ell \log k}\displaystyle\sum_{x \in \Sigma^{\ell}}\pi_{u_\ell}^{(\ell)}(x)\bigg[\log \frac{1}{\pi_{u_\ell}^{(\ell)}(x)} - \log 2^{L_C(x)} \bigg] + \Big \lfloor \frac{n}{\ell}  \Big \rfloor^{-1} \\
&= \frac{1}{\ell \log k}\displaystyle\sum_{x \in \Sigma^{\ell}}\pi_{u_\ell}^{(\ell)}(x)\log \frac{2^{-L_C(x)}}{\pi_{u_\ell}^{(\ell)}(x)} + \Big \lfloor \frac{n}{\ell}  \Big \rfloor^{-1}.  \nonumber
\end{align}
By (\ref{firstineq}), Jensen's Inequality, and the Generalized Kraft Inequality,
\begin{align*}
\frac{\ent(\pi_{u_\ell}^{(\ell)})}{\ell \log k} - \rho_C(u) &\leq \frac{1}{\ell \log k}\log \bigg ( \displaystyle\sum_{x \in \Sigma^{\ell}}2^{-L_C(x)} \bigg ) + \Big \lfloor \frac{n}{\ell}  \Big \rfloor^{-1}\\
&\leq \frac{1}{\ell \log k}\log \bigg (s^2 \bigg (1 + \log{\frac{s^2 + k^{\ell}}{s^2}} \bigg )  \bigg ) + \Big \lfloor \frac{n}{\ell}  \Big \rfloor^{-1}\\
&= \Big \lfloor \frac{n}{\ell}  \Big \rfloor^{-1} + f_s^k(\ell). \qedhere
\end{align*}
\end{proof}

The following lemma is a ``joint'' version of Lemma \ref{low}.

\begin{lem}\label{low2}
Let $C$ be an ILFSC on $\Sigma \times \Sigma$. For every $\ell,n \in \mathbb{Z}^+$ and every $u,w \in \Sigma^n$ such that $n$ is a multiple of $\ell$,
\[
\rho_C(u,w) \geq \frac{1}{\ell \log k}\displaystyle\sum_{(x,y) \in \Sigma^{\ell}}\pi_{u,w}^{(\ell)}(x,y)L_C((x,y)).
\]
\end{lem}

\begin{proof}
First, recall that that, for all $x,y \in \Sigma^{\ell}$, $u_{(u,w)}^{(\ell)}((x,y))$ = $u_{u,w}^{(\ell)}(x,y)$. By Lemma \ref{low},
\begin{align*}
\rho_C(u,w) &= \frac{2}{2}\rho_C(u,w)\\
						&=2\rho_C((u,w))\\
						&\geq 2\frac{1}{\ell \log k^2}\displaystyle\sum_{(x,y) \in \Sigma^{\ell}}\pi_{(u,w)}^{(\ell)}((x,y))L_C((x,y))\\
						&= \frac{1}{\ell \log k}\displaystyle\sum_{(x,y) \in \Sigma^{\ell}}\pi_{u,w}^{(\ell)}(x,y)L_C((x,y)). \qedhere
\end{align*}
\end{proof}

\begin{lem}\label{hc2}
Let $C$ be an ILFSC on $\Sigma \times \Sigma$ with $s \in \mathbb{Z}^+$ states. For every $\ell,n \in \mathbb{N}$ and every $u \in \Sigma^n$ and $w \in \Sigma^n$ such that $\ell \leq n$,
\[
\frac{\emph{\ent}(\pi_{u_\ell,w_\ell}^{(\ell)})}{\ell \log k} - \rho_C(u,w) \leq \Big \lfloor \frac{n}{\ell}  \Big \rfloor^{-1} + f_s^{k^2}(\ell),
\]
where $u_\ell = u \harp \Big \lfloor \frac{n}{\ell}  \Big \rfloor \cdot \ell$, $w_\ell = w \harp \Big \lfloor \frac{n}{\ell}  \Big \rfloor \cdot \ell$, and $\lim_{m \rightarrow \infty}f_s^{k^2}(m) = 0$.
\end{lem}

\begin{proof}
The proof of this lemma is identical to the proof of Lemma \ref{hc}, except that it uses the second inequality from Observation \ref{obs1} instead of the first inequality and it also uses Lemma \ref{low2} instead of Lemma \ref{low}.
\end{proof}

If $|u|$ is a multiple of $\ell$, we denote $C_{F(\ell,u)}$ as the ILFSC on $\Sigma$ consisting of $k^{\ell}$ states that encodes strings of length $\ell$ according to Huffman's algorithm on the frequencies of each string in $\Sigma^{\ell}$ within $u$. The following lemma is an inequality that was noted by Sheinwald in \cite{jShei94}.

\begin{lem}[Sheinwald \cite{jShei94}]\label{huff}
Let $C$ be an ILFSC on $\Sigma$. For every $\ell,n \in \mathbb{Z}^+$ and $u \in \Sigma^n$ such that $n$ is a multiple of $\ell$,
\[
\rho_{C_{F(\ell,u)}}(u) \leq \frac{\emph{\ent}(\pi_u^{(\ell)})}{\ell \log k} + \frac{1}{\ell}.
\]
\end{lem}

\begin{lem}\label{huffc}
For each $r,n \in \mathbb{Z}^+$ and each $u \in \Sigma^n$,
\[
\rho_{r}(u) \leq \frac{\emph{\ent}(\pi_{u_{r'}}^{(r')})}{r' \log k} + \frac{1}{r'},
\]
where $r' = \lfloor log_k r \rfloor$ and $u_{r'} = u \harp \Big \lfloor \frac{n}{r'} \Big \rfloor \cdot r'$.
\end{lem}

\begin{proof}
Since, $\rho_{k^{r'}}(u) \leq \rho_{F(r',u_{r'})}(u)$, we have
\begin{align}\label{f}
\rho_r(u) &= \rho_{k^{log_k r}}(u) \nonumber \\
					&\leq \rho_{k^{r'}}(u)\\
					&\leq \rho_{C_{F(r',u_{r'})}}(u). \nonumber
\end{align}
By (\ref{f}), and since $|C_{F(r',u_{r'})}(u)| = |C_{F(r',u_{r'})}(u_{r'})|$, we have
\begin{align}\label{secondp}
\rho_r(u) &\leq \frac{|C_{F(r',u_{r'})}(u)|}{n\log k} \nonumber \\
					&= \frac{|C_{F(r',u_{r'})}(u_{r'})|}{n\log k}\\
					&\leq \frac{|C_{F(r',u_{r'})}(u_{r'})|}{\Big \lfloor \frac{n}{r'} \Big \rfloor \cdot r'\log k} \nonumber \\
					&= \rho_{C_{F(r',u_{r'})}}(u_{r'}). \nonumber
\end{align}
Finally, by (\ref{secondp}) and Lemma \ref{huff}, we have
\begin{align*}
\rho_r(u)	&\leq \frac{\ent(u_{r'}^{(r')})}{r' \log k} + \frac{1}{r'}. \qedhere
\end{align*}
\end{proof}

\begin{lem}\label{maxin}
For each $r,t,n \in \mathbb{Z}^+$ and $(u,w) \in (\Sigma \times \Sigma)^n$ such that $n \geq t'$,
\[
\max\{\rho_t(u),\rho_t(w)\} \leq \rho_r(u,w) + \Big \lfloor \frac{n}{t'}  \Big \rfloor^{-1} + g_r^k(t'),
\]
where $t' = \lfloor \log_k t \rfloor$ and $\displaystyle\lim_{m \rightarrow \infty}g_r^k(m) = 0$.
\end{lem}

\begin{proof}
By Lemma \ref{huffc} and Corollary \ref{Shac}, we observe that
\begin{align}\label{hfirst}
\rho_t(u) - \rho_r(u,w) &\leq \frac{\ent(\pi_{u_{t'}}^{(t')})}{t' \log k} - \rho_r(u,w) + \frac{1}{t'} \nonumber \\
															&\leq \frac{\ent(\pi_{u_{t'},w_{t'}}^{(t')})}{t' \log k} - \rho_r(u,w) + \frac{1}{t'},
\end{align}
where $t' = \lfloor \log_k t \rfloor$. Finally, by (\ref{hfirst}) and Lemma \ref{hc2}, we have
\begin{align*}
\rho_t(u) - \rho_r(u,w) &\leq \frac{\ent(\pi_{u_{t'},w_{t'}}^{(t')})}{t' \log k} - \frac{\ent(\pi_{u_{t'},w_{t'}}^{(t')})}{t' \log k} + \Big \lfloor \frac{n}{t'} \Big \rfloor^{-1} + f_r^{k^2}(t') + \frac{1}{t'}\nonumber \\
												&\leq \Big \lfloor \frac{n}{t'}  \Big \rfloor^{-1} + f_r^{k^2}(t') + \frac{1}{t'},
\end{align*}
which immediately implies that
\begin{align*}
\rho_t(u) &\leq \rho_r(u,w) + \Big \lfloor \frac{n}{t'}  \Big \rfloor^{-1} + f_r^{k^2}(t') + \frac{1}{t'}\\
					&= \rho_r(u,w) + \Big \lfloor \frac{n}{t'}  \Big \rfloor^{-1} + g_r^k(t'),
\end{align*}
where $g_r^k(t') = f_r^{k^2}(t') + \frac{1}{t'}$. Similarly, we can show that
\[
\rho_t(w) \leq \rho_r(u,w) + \Big \lfloor \frac{n}{t'}  \Big \rfloor^{-1} + g_r^k(t'),
\]
which implies the conclusion. \qedhere
\end{proof}

\begin{lem}\label{1to2}
For every $r,t,n \in \mathbb{Z}^+$ and every $(u,w) \in (\Sigma \times \Sigma)^n$ such that $n \geq r'$,
\[
\rho_r(u,w) \leq \rho_t(u) + \rho_t(w) + 2\Big \lfloor \frac{n}{r'}  \Big \rfloor^{-1} + h_t^k(r').
\]
where $r' = \lfloor \log_k r \rfloor$ and $\displaystyle\lim_{m \rightarrow \infty}h_t^k(m) = 0$.
\end{lem}

\begin{proof}
By Lemma \ref{huffc} and Corollary \ref{Shac}, we observe that
\begin{align}\label{fapp}
\rho_t(u) + \rho_t(w) - \rho_r(u,w) &= \rho_t(u) + \rho_t(w) - 2\rho_r((u,w)) \nonumber\\
																					 &\geq \rho_t(u) + \rho_t(w) - 2\bigg(\frac{\ent\big(\pi_{(u_{r'},w_{r'})}^{(r')}\big)}{r' \log k^2} + \frac{1}{r'}\bigg) \\
																					 &= \rho_t(u) + \rho_t(w) - \frac{\ent(\pi_{u_{r'},w_{r'}}^{(r')})}{r' \log k} - \frac{2}{r'} \nonumber\\
																					 &\geq \rho_t(u) + \rho_t(w) - \frac{\ent(\pi_{u_{r'}}^{(r')})}{r' \log k} - \frac{\ent(\pi_{w_{r'}}^{(r')})}{r' \log k} - \frac{2}{r'}. \nonumber
\end{align}
Finally, by (\ref{fapp}) and Lemma \ref{hc}, we have
\begin{align*}
\rho_t(u) + \rho_t(w) - \rho_r(u,w) &\geq -2\Big \lfloor \frac{n}{r'}  \Big \rfloor^{-1} - 2f_t^k(r') - \frac{2}{r'}\\
																		&= -2\Big \lfloor \frac{n}{r'}  \Big \rfloor^{-1} - h_t^k(r'),
\end{align*}
where $h_t^k(r') = 2f_t^k(r') + \frac{2}{r'}$, which immediately implies the conclusion. \qedhere
\end{proof}

\begin{lem}\label{eq}
For each $r,t,n \in \mathbb{Z}^+$ and $u \in \Sigma^n$ such that $n \geq \max\{r',t'\}$,
\[
\rho_{r}(u,u) \leq \rho_t(u) + \Big \lfloor \frac{n}{r'}  \Big \rfloor^{-1} + i_t^k(r')
\]
and
\[
\rho_{t}(u) \leq \rho_{r}(u,u) + \Big \lfloor \frac{n}{t'} \Big \rfloor^{-1} + j_r^k(t'),
\]
where $r' = \lfloor log_k r \rfloor$, $t' = \lfloor log_k t \rfloor$, $\displaystyle\lim_{m \rightarrow \infty}i_t^k(m) = 0$, and $\displaystyle\lim_{m \rightarrow \infty}j_r^k(m) = 0$.
\end{lem}

\begin{proof}
By Lemma \ref{huffc} and Corollary \ref{Shac} we know that
\begin{align}\label{firstp}
\rho_{r}(u,u) - \rho_t(u) &= 2\rho_{r}((u,u)) - \rho_t(u) \nonumber \\
																					 &\leq 2\bigg(\frac{\ent\big(\pi_{(u_{r'},u_{r'})}^{(r')}\big)}{r' \log k^2} + \frac{1}{r'} \bigg ) - \rho_t(u) \nonumber\\
																					 &= \frac{\ent(\pi_{u_{r'},u_{r'}}^{(r')})}{r' \log k} + \frac{2}{r'} - \rho_t(u)\\
																					 &= \frac{\ent(\pi_{u_{r'}}^{(r')})}{r' \log k} - \rho_t(u)
 + \frac{2}{r'}. \nonumber
\end{align}
By (\ref{firstp}) and Lemma \ref{hc}, we have
\begin{align*}
\rho_{r}(u,u) - \rho_t(u) &\leq \Big \lfloor \frac{n}{r'}  \Big \rfloor^{-1} + f_t^k(r') + \frac{2}{r'}\\
													&= \Big \lfloor \frac{n}{r'} \Big \rfloor^{-1} + i_t^k(r'),
\end{align*}
where $i_t^k(r') = f_t^k(r') + \frac{2}{r'}$, which implies the first inequality. Now, by Lemma \ref{hc2}, Lemma \ref{huffc}, and Corollary \ref{Shac}, we know that
\begin{align*}
\rho_r(u,u) - \rho_{t}(u) &\geq \frac{\ent(\pi_{u_{t'},u_{t'}}^{(t')})}{t' \log k} - \rho_{t}(u) - \Big \lfloor \frac{n}{t'}  \Big \rfloor^{-1} - f_r^{k^2}(t')\\
															&\geq \frac{\ent(\pi_{u_{t'},u_{t'}}^{(t')})}{t' \log k} - \frac{\ent(\pi_{u_{t'}}^{(t')})}{t' \log k} - \Big \lfloor \frac{n}{t'}  \Big \rfloor^{-1} - f_r^{k^2}(t') - \frac{1}{t'}\\
															&= \frac{\ent(\pi_{u_{t'}}^{(t')})}{t' \log k} - \frac{\ent(\pi_{u_{t'}}^{(t')})}{t' \log k} - \Big \lfloor \frac{n}{t'}  \Big \rfloor^{-1} - f_r^{k^2}(t') - \frac{1}{t'}\\
															&= - \Big \lfloor \frac{n}{t'}  \Big \rfloor^{-1} - j_r^k(t'),
\end{align*}
where $j_r^k(t') = f_r^{k^2}(t') + \frac{1}{t'}$, which implies the second inequality. \qedhere
\end{proof}

\begin{obs}\label{obs2}
For any ILFSC $C$ on $\Sigma \times \Sigma$ and any $u \in \Sigma^n$ and $w \in \Sigma^n$, there exists another ILFSC $C'$ such that
\[
\rho_C(u,w) = \rho_{C'}(w,u).
\]
and $C$ and $C'$ have the same number of states.
\end{obs}

\begin{proof}
It is clear that we may relabel the input symbol $(a,b) \in \Sigma \times \Sigma$ along each transition of $C$ by interchanging the symbols within the pair. So each each $(a,b)$-transition becomes a $(b,a)$-transition. After performing this relabeling, we receive an ILFSC $C'$ on $\Sigma \times \Sigma$ with the same number of states as $C$. Since the output string along each transition has not changed, $C$ compresses $(u,w)$ just as well as $C'$ compresses $(w,u)$. \qedhere 
\end{proof}

\begin{lem}\label{inter}
For every $r,n \in \mathbb{Z}^+$ and every $u \in \Sigma^n$ and $w \in \Sigma^n$,
\[
\rho_r(u,w) = \rho_r(w,u).
\]
\end{lem}

\begin{proof}
let $C$ be the ILFSC on $\Sigma \times \Sigma$ such that $\rho_C(u,w) = \rho_r(u,w)$. By Observation \ref{obs2}, there exists another ILFSC $C'$ on $\Sigma \times \Sigma$ such that $\rho_C(u,w) = \rho_{C'}(w,u)$ and $C$ and $C'$ have the same number of states. Therefore,
\begin{align*}
\rho_r(u,w) &= \rho_{C}(u,w)\\
						&= \rho_{C'}(w,u)\\
						&\geq \rho_r(w,u).
\end{align*}
Using similar reasoning, we can show that $\rho_r(u,w) \leq \rho_r(w,u)$. \qedhere
\end{proof}

We proceed to explore finite-state mutual compression ratios between strings.

\begin{defn}
Let $r,t \in \mathbb{Z}^+$. The $r,t$-\emph{state mutual compression ratio} between $u \in \Sigma^n$ and $w \in \Sigma^n$ is
\[
\rho_{r,t}(u:w) = \rho_t(u) + \rho_t(w) - \rho_r(u,w).
\]
\end{defn}

We now present the main theorem of this section, which lists the basic properties of mutual compression ratios for finite-length strings.

\begin{thm}[Properties of Mutual Compression Ratios between Strings]\label{mcr}
For every $r,t,n \in \mathbb{Z}^+$ and every $u \in \Sigma^n$ and $w \in \Sigma^n$ such that $n \geq \max\{t',r'\}$,
\begin{enumerate}
\item $\rho_{r,t}(u:w) \leq \min\{\rho_{t}(u),\rho_{t}(w)\} + \Big\lfloor \frac{n}{t'} \Big\rfloor^{-1} + g_r^k(t')$ and $\displaystyle\lim_{m \rightarrow \infty}g_r^k(m) = 0$,
\item $\rho_{r,t}(u:w) + 2\Big\lfloor \frac{n}{r'} \Big\rfloor^{-1} + h_t^k(r') \geq 0$ and $\displaystyle\lim_{m \rightarrow \infty}h_t^k(m) = 0$,
\item $\rho_{r,t}(u:u) + \Big\lfloor \frac{n}{r'} \Big\rfloor^{-1} + i_t^k(r') \geq \rho_t(u)$ and $\displaystyle\lim_{m \rightarrow \infty}i_t^k(m) = 0$,
\item $\rho_{r,t}(u:u) \leq \rho_t(u) + \Big\lfloor \frac{n}{t'} \Big\rfloor^{-1} + j_r^k(t')$ and $\displaystyle\lim_{m \rightarrow \infty}j_r^k(m) = 0$,
\item $\rho_{r,t}(u:w) = \rho_{r,t}(w:u)$, and
\item $\rho_{r,t}(u:w) \leq \rho_{t,r}(u:w) + 3\Big \lfloor \frac{n}{t'} \Big \rfloor^{-1} + e_r^k(t')$ and $\displaystyle\lim_{m \rightarrow \infty}e_r^k(m) = 0$,
\end{enumerate}
where $r' = \lfloor log_k r \rfloor$ and $t' = \lfloor log_k t \rfloor$.
\end{thm}

\begin{proof}
By Lemma \ref{maxin},
\begin{align*}
\rho_{r,t}(u:w) &= \rho_{t}(u) + \rho_{t}(w) - \rho_r(u,w)\\
									&\leq \min\{\rho_{t}(u), \rho_{t}(w)\} + \bigg\lfloor \frac{n}{t'} \bigg\rfloor^{-1} + g_r^k(t'),
\end{align*}
which proves the first statement. By Lemma \ref{1to2},
\begin{align*}
\rho_{r,t}(u:w) &= \rho_t(u) + \rho_t(w) - \rho_{r}(u,w)\\
								&\geq - 2\Big\lfloor \frac{n}{r'} \Big\rfloor^{-1} - h_t^k(r'),
\end{align*}
which proves the second statement.  By the first inequality of Lemma \ref{eq},
\begin{align*}
\rho_{r,t}(u:u) &= \rho_{t}(u) + \rho_{t}(u) - \rho_r(u,u)\\
									 &\geq \rho_{t}(u) - \Big\lfloor \frac{n}{r'} \Big\rfloor^{-1} - i_t^k(r'),
\end{align*}
which proves the third statement. By the second inequality of Lemma \ref{eq},
\begin{align*}
\rho_{r,t}(u:u) &= \rho_t(u) + \rho_t(u) - \rho_{r}(u,u)\\
										 &\leq \rho_t(u) + \Big \lfloor \frac{n}{t'} \Big \rfloor^{-1} + j_r^k(t'),
\end{align*}
which proves the fourth statement. By Lemma \ref{inter},
\begin{align*}
\rho_{r,t}(u:w) &= \rho_t(u) + \rho_t(w) - \rho_{r}(u,w)\\
							  &= \rho_t(w) + \rho_t(u) - \rho_{r}(w,u)\\
								&= \rho_{r,t}(w:u),
\end{align*}
which proves the fifth statement. Finally, to prove the sixth statement, observe that
\begin{align*}
\rho_{r,t}(u:w) - \rho_{t,r}(u:w) 	&= \rho_t(u) + \rho_t(w) - \rho_r(u, w) - \rho_r(u) - \rho_r(w) + \rho_t(u,w)\\
&= \rho_t(u) + \rho_t(w) - \rho_r(u,w) - \rho_r(u) - \rho_r(w) + 2\rho_t((u,w))
\end{align*}
By the above inequality and Lemmas \ref{hc}, \ref{hc2}, and \ref{huffc},
\begin{align*}
&\rho_{r,t}(u:w) - \rho_{t,r}(u:w)\\
&\leq \frac{\ent(\pi_{u_{t'}}^{(t')})}{t' \log k} + \frac{\ent(\pi_{w_{t'}}^{(t')})}{t' \log k} - \frac{\ent(\pi_{u_{t'},w_{t'}}^{(t')})}{t' \log k} - \frac{\ent(\pi_{u_{t'}}^{(t')})}{t' \log k} - \frac{\ent(\pi_{w_{t'}}^{(t')})}{t' \log k} + 2\bigg(\frac{\ent\big(\pi_{(u,w)_{t'}}^{(t')}\big)}{t' \log k^2} + \frac{1}{t'}\bigg)\\
&\hspace*{10mm} + \frac{2}{t'} + 3\Big \lfloor \frac{n}{t'} \Big \rfloor^{-1} + 2f_r^{k}(t') + f_r^{k^2}(t')\\
&= \frac{\ent(\pi_{u_{t'}}^{(t')})}{t' \log k} + \frac{\ent(\pi_{w_{t'}}^{(t')})}{t' \log k} - \frac{\ent(\pi_{u_{t'},w_{t'}}^{(t')})}{t' \log k} - \frac{\ent(\pi_{u_{t'}}^{(t')})}{t' \log k} - \frac{\ent(\pi_{w_{t'}}^{(t')})}{t' \log k} + \frac{\ent(\pi_{u_{t'},w_{t'}}^{(t')})}{t' \log k}\\
&\hspace*{10mm} + \frac{4}{t'} + 3\Big \lfloor \frac{n}{t'} \Big \rfloor^{-1} + 2f_r^{k}(t') + f_r^{k^2}(t')\\
&= \frac{4}{t'} + 3\Big \lfloor \frac{n}{t'} \Big \rfloor^{-1} + 2f_r^{k}(t') + f_r^{k^2}(t')\\
&= 3\Big \lfloor \frac{n}{t'} \Big \rfloor^{-1} + e_r^k(t'),
\end{align*}
where $e_r^k(t') = \frac{4}{t'} + 2f_r^{k}(t') + f_r^{k^2}(t')$. \qedhere
\end{proof}


\section{Finite-State Mutual Dimension}
In this section we define the \emph{lower} and \emph{upper mutual compression ratios} and the \emph{lower} and \emph{upper finite-state mutual dimensions} between sequences and explore their properties.

We begin by discussing the \emph{finite-state dimension} $dim_{FS}(S)$ of a sequence $S \in \Sigma^{\infty}$, which was originally defined in 2003 by Dai, Lathrop, Lutz, and Mayordomo in \cite{jDaLaLuMa04} using finite-state gamblers. In the same paper, the authors proved a characterization of finite-state dimension using finite-state compressors. In 2007, Athreya, Hitchcock, Lutz, and Mayordomo defined the \emph{finite-state strong dimension} of a sequence using finite-state gamblers and proved that it can also be characterized using finite-state compressors \cite{jAHLM07}. In this section, we will use the compressor characterization of finite-state dimension and finite-state strong dimension and refer to them as the \emph{lower} and \emph{upper finite-state dimensions}, respectively.

\begin{defn}
Let $r\in \mathbb{Z}^+$ and $S \in \Sigma^\infty$. The \emph{lower} and \emph{upper} $r$-\emph{state compression ratios} of $S$ are
\[
\rho_r(S) = \displaystyle\liminf_{n \rightarrow \infty}\rho_r(S \harp n)
\]
and
\[
\hat{\rho}_r(S) = \displaystyle\limsup_{n \rightarrow \infty}\rho_r(S \harp n),
\]
respectively.
\end{defn}

\begin{defn}
Let $r\in \mathbb{Z}^+$ and $S,T \in \Sigma^\infty$. The \emph{lower} and \emph{upper} $r$-\emph{state joint compression ratios} of $S$ and $T$ are
\[
\rho_r(S,T) = \displaystyle\liminf_{n \rightarrow \infty}\rho_r(S \harp n,T \harp n)
\]
and
\[
\hat{\rho}_r(S,T) = \displaystyle\limsup_{n \rightarrow \infty}\rho_r(S \harp n,T \harp n),
\]
respectively.
\end{defn}

\begin{defn}
Let $r,t \in \mathbb{Z}^+$. The \emph{lower} and \emph{upper} $r,t$-\emph{state mutual compression ratios} between $S \in \Sigma^{\infty}$ and $T \in \Sigma^{\infty}$ are
\[
\rho_{r,t}(S:T) = \displaystyle\liminf_{n \rightarrow \infty}\rho_{r,t}(S \harp n:T \harp n)
\]
and
\[
\hat{\rho}_{r,t}(S:T) = \displaystyle\limsup_{n \rightarrow \infty}\rho_{r,t}(S \harp n:T \harp n),
\]
respectively.
\end{defn}

We now present and prove the properties of the lower and upper $r,t$-state mutual compression ratio between sequences.

\begin{lem}[Properties of Mutual Compression Ratios between Sequences]\label{mcrs}
Let $r,t \in \mathbb{Z}^+$. For all $S,T \in \Sigma^\infty$,
{\footnotesize
\begin{enumerate}
\item $\rho_t(S) + \rho_t(T) - \hat{\rho}_r(S,T) \leq \rho_{r,t}(S:T) \leq \hat{\rho}_t(S) + \hat{\rho}_t(T) - \hat{\rho}_r(S,T)$
\item $\rho_t(S) + \rho_t(T) - \rho_r(S,T) \leq \hat{\rho}_{r,t}(S:T) \leq \hat{\rho}_t(S) + \hat{\rho}_t(T) - \rho_r(S,T)$
\item $\rho_{r,t}(S:T) \leq \min\{\rho_t(S),\rho_t(T)\} + g_r^k(t')$, $\hat{\rho}_{r,t}(S:T) \leq \min\{\hat{\rho}_{t}(S),\hat{\rho}_{t}(T)\} + g_r^k(t')$,\\ and $\displaystyle\lim_{m \rightarrow \infty}g_r^k(m) = 0$,
\item $\rho_{r,t}(S:T) + h_t^k(r') \geq 0$, $\hat{\rho}_{r,t}(S:T) + h_t^k(r') \geq 0$, and $\displaystyle\lim_{m \rightarrow \infty}h_t^k(m) = 0$,
\item $\rho_{r,t}(S:S) + i_t^k(r') \geq \rho_t(S)$, $\hat{\rho}_{r,t}(S:S) + i_t^k(r') \geq \hat{\rho}_t(S)$, and $\displaystyle\lim_{m \rightarrow \infty}i_t^k(m) = 0$,
\item $\rho_{r,t}(S:S) \leq \rho_t(S) + j_r^k(t')$, $\hat{\rho}_{r,t}(S:S) \leq \hat{\rho}_t(S) + j_r^k(t')$, and $\displaystyle\lim_{m \rightarrow \infty}j_r^k(m) = 0$,
\item $\rho_{r,t}(S:T) = \rho_{r,t}(T:S)$, $\hat{\rho}_{r,t}(S:T) = \hat{\rho}_{r,t}(T:S)$, and
\item $\rho_{r,t}(S:T) \leq \rho_{t,r}(S:T) + e_r^k(t')$, and $\displaystyle\lim_{m \rightarrow \infty}e_r^k(m) = 0$,
\end{enumerate}
}
where $r' = \lfloor log_k r \rfloor$ and $t' = \lfloor log_k t \rfloor$.
\end{lem}

\begin{proof}
To prove the first inequality in the first statement, observe that
\begin{align*}
\rho_{r,t}(S:T) &= \displaystyle\liminf_{n \rightarrow \infty}\rho_{r,t}(S \harp n:T \harp n)\\
								&= \displaystyle\liminf_{n \rightarrow \infty}\Big[\rho_t(S \harp n) + \rho_t(T \harp n) - \rho_r(S \harp n, T \harp n)\Big]\\
								&\geq \displaystyle\liminf_{n \rightarrow \infty}\rho_t(S \harp n) + \displaystyle\liminf_{n \rightarrow \infty}\rho_t(T \harp n) + \displaystyle\liminf_{n \rightarrow \infty}-\rho_r(S \harp n, T \harp n)\\
								&= \displaystyle\liminf_{n \rightarrow \infty}\rho_t(S \harp n) + \displaystyle\liminf_{n \rightarrow \infty}\rho_t(T \harp n) - \displaystyle\limsup_{n \rightarrow \infty}\rho_r(S \harp n, T \harp n)\\
								&= \rho_t(S) + \rho_t(T) - \hat{\rho_r}(S,T).
\end{align*}
For the second inequality, observe that
\begin{align*}
&\hat{\rho}_t(S) + \hat{\rho}_t(T) - \rho_{r,t}(S:T)\\
&= \displaystyle\limsup_{n \rightarrow \infty}\rho_t(S \harp n) + \displaystyle\limsup_{n \rightarrow \infty}\rho_t(T \harp n) - \displaystyle\liminf_{n \rightarrow \infty}\rho_{r,t}(S \harp n:T \harp n)\\
&= \displaystyle\limsup_{n \rightarrow \infty}\rho_t(S \harp n) + \displaystyle\limsup_{n \rightarrow \infty}\rho_t(T \harp n) + \displaystyle\limsup_{n \rightarrow \infty}-\rho_{r,t}(S \harp n:T \harp n)\\
&\geq \displaystyle\limsup_{n \rightarrow \infty}\Big[\rho_t(S \harp n) + \rho_t(T \harp n) - \rho_{r,t}(S \harp n:T \harp n)\Big]\\
&= \displaystyle\limsup_{n \rightarrow \infty}\rho_{r}(S \harp n,T \harp n)\\
&= \hat{\rho_r}(S,T),
\end{align*}
which implies the second inequality of the first statement. The second statement has a similar proof as the first statement. The third statement follows from the first statement of Theorem \ref{mcr} and the fact that, for any $t' \in \mathbb{Z}^+$, $\displaystyle\lim_{n \rightarrow \infty}\lfloor n/t' \rfloor^{-1} = 0$. The fourth statement follows from the second statement of Theorem \ref{mcr} and the fact that, for any $r' \in \mathbb{Z}^+$, $\displaystyle\lim_{n \rightarrow \infty}2\lfloor n/r' \rfloor^{-1} = 0$. The fifth statement follows from the third statement of Theorem \ref{mcr} and the fact that, for any $r' \in \mathbb{Z}^+$, $\displaystyle\lim_{n \rightarrow \infty}\lfloor n/r' \rfloor^{-1} = 0$. The sixth statement follows from the fourth statement of Theorem \ref{mcr} and the fact that, for all $t' \in \mathbb{Z}^+$, $\displaystyle\lim_{n \rightarrow \infty}\lfloor n/t' \rfloor^{-1} = 0$. The seventh statement follows directly from the fifth statement of Theorem \ref{mcr}. Finally, to prove the eighth statement, observe that, by the sixth statement of Theorem \ref{mcr} and by the fact that, for all $t' \in \mathbb{Z}^+$, $\displaystyle\lim_{n \rightarrow \infty}3\lfloor n/t' \rfloor^{-1} = 0$,
\begin{align*}
\rho_{r,t}(S:T) - \rho_{t,r}(S:T) &= \displaystyle\liminf_{n \rightarrow \infty}\rho_{r,t}(S \harp n:T \harp n) - \displaystyle\liminf_{n \rightarrow \infty}\rho_{t,r}(S \harp n:T \harp n)\\
																	&\leq \displaystyle\limsup_{n \rightarrow \infty}\big[ \rho_{r,t}(S \harp n:T \harp n) - \rho_{t,r}(S \harp n:T \harp n) \big]\\
																	&= \displaystyle\limsup_{n \rightarrow \infty}\bigg[3\Big \lfloor \frac{n}{t'} \Big \rfloor^{-1} + e_r^k(t')\bigg ]\\
																	&= \displaystyle\limsup_{n \rightarrow \infty}\bigg[ e_r^k(t')\bigg ]\\
																	&= e_r^k(t'). \qedhere
\end{align*}
\end{proof}

We proceed to discuss the compression ratio characterization of finite-state dimension.

\begin{defn}
The \emph{lower} and \emph{upper finite-state compression ratios} of $S \in \Sigma^{\infty}$ are
\[
\rho(S) = \displaystyle\lim_{r \rightarrow \infty}\liminf_{n \rightarrow \infty}\rho_r(S \harp n)
\]
and
\[
\hat{\rho}(S) = \displaystyle\lim_{r \rightarrow \infty}\limsup_{n \rightarrow \infty}\rho_r(S \harp n),
\]
respectively.
\end{defn}
\noindent (Note that, by the monotone convergence theorem, the definitions of the lower and upper finite-state compression ratios are equal to those found in (\ref{compratio}) and (\ref{scompratio}), respectively, since $\rho_r(S \harp n)$ is bounded and decreasing in $r$.)
\begin{defn}
The \emph{lower} and \emph{upper joint finite-state compression ratios} of $S \in \Sigma^{\infty}$ are
\[
\rho(S,T) = \displaystyle\lim_{r \rightarrow \infty}\liminf_{n \rightarrow \infty}\rho_r(S \harp n,T \harp n)
\]
and
\[
\hat{\rho}(S,T) = \displaystyle\lim_{r \rightarrow \infty}\limsup_{n \rightarrow \infty}\rho_r(S \harp n,T \harp n),
\]
respectively.
\end{defn}

In the following theorem, the first equality was proven by Dai, Lathrop, Lutz, and Mayordomo in \cite{jDaLaLuMa04} and the second equality was proven by Athreya, Hitchcock, Lutz, and Mayordomo in \cite{jAHLM07}.

\begin{thm}[\cite{jDaLaLuMa04, jAHLM07}]\label{dimtorho}
For all $S,T \in \Sigma^\infty$,
\[
dim_{FS}(S) = \rho(S)
\]
and
\[
Dim_{FS}(S) = \hat{\rho}(S).
\]
\end{thm}

\noindent The following corollary follows directly from Theorem \ref{dimtorho}.

\begin{cor}\label{jdimtorho}
For all $S,T \in \Sigma^\infty$,
\[
dim_{FS}(S,T) = \rho(S,T)
\]
and
\[
Dim_{FS}(S,T) = \hat{\rho}(S,T).
\]
\end{cor}

We now present the definitions of the \emph{lower} and \emph{upper finite-state mutual dimensions} between sequences.

\begin{defn}
The \emph{lower} and \emph{upper finite-state mutual dimensions} between $S \in \Sigma^{\infty}$ and $T \in \Sigma^{\infty}$ are
\[
mdim_{FS}(S:T) = \displaystyle\lim_{r \rightarrow \infty}\lim_{t \rightarrow \infty}\rho_{r,t}(S:T)
\]
and
\[
Mdim_{FS}(S:T) = \displaystyle\lim_{r \rightarrow \infty}\lim_{t \rightarrow \infty}\hat{\rho}_{r,t}(S:T),
\]
respectively.
\end{defn}

The first limit in the definitions above exists because both $\rho_{r,t}(S:T)$ and $\hat{\rho}_{r,t}(S:T)$ are decreasing in $t$ since $\rho_t(S \harp n)$ and $\rho_t(T \harp n)$ are decreasing in $t$. The second limit also exists because both
\[
\displaystyle\lim_{t \rightarrow \infty}\rho_{r,t}(S:T) \text{ and } \displaystyle\lim_{t \rightarrow \infty}\hat{\rho}_{r,t}(S:T)
\] 
are increasing in $r$, since $-\rho_{r}(S \harp n, T \harp n)$ is increasing in $r$.

Our first theorem of this section is an important result that allows for the interchanging of the iterated limits within the definitions of the lower and upper finite-state mutual dimensions. The proof of the properties of finite-state mutual dimensions (Theorem \ref{md}) rely on this result.

\begin{thm}\label{interchange}
For all $S,T \in \Sigma^{\infty}$,
\[
mdim_{FS}(S:T) = \displaystyle\lim_{r \rightarrow \infty}\lim_{t \rightarrow \infty}\rho_{r,t}(S:T) = \displaystyle\lim_{t \rightarrow \infty}\lim_{r \rightarrow \infty}\rho_{r,t}(S:T)
\]
and
\[
Mdim_{FS}(S:T) = \displaystyle\lim_{r \rightarrow \infty}\lim_{t \rightarrow \infty}\hat{\rho}_{r,t}(S:T) = \displaystyle\lim_{t \rightarrow \infty}\lim_{r \rightarrow \infty}\hat{\rho}_{r,t}(S:T).
\]
\end{thm}

\begin{proof}
Observe that the iterated limits
\[
\displaystyle\lim_{t \rightarrow \infty}\lim_{r \rightarrow \infty}\rho_{r,t}(S:T) \text{ and } \displaystyle\lim_{t \rightarrow \infty}\lim_{r \rightarrow \infty}\hat{\rho}_{r,t}(S:T)
\]
exist for the same reason the iterated limits exist in the original definitions. First, we show that
\[
mdim_{FS}(S:T) \leq \displaystyle\lim_{t \rightarrow \infty}\lim_{r \rightarrow \infty}\rho_{r,t}(S:T).
\]
By the eighth statement of Lemma \ref{mcrs},
\begin{align*}
mdim_{FS}(S:T) &= \displaystyle\lim_{r \rightarrow \infty}\lim_{t \rightarrow \infty}\rho_{r,t}(S:T)\\
							 &\leq \displaystyle\lim_{r \rightarrow \infty}\lim_{t \rightarrow \infty}\bigg[\rho_{t,r}(S:T) + e_r^k(t')\bigg ]\\
							 &=\displaystyle\lim_{r \rightarrow \infty}\lim_{t \rightarrow \infty}\rho_{t,r}(S:T)\\
							 &= \displaystyle\lim_{t \rightarrow \infty}\lim_{r \rightarrow \infty}\rho_{r,t}(S:T).
\end{align*}
Next, we show that
\[
mdim_{FS}(S:T) \geq \displaystyle\lim_{t \rightarrow \infty}\lim_{r \rightarrow \infty}\rho_{r,t}(S:T).
\]
Let $\epsilon, \epsilon_1, \epsilon_2, \epsilon_3 > 0$ such that $\epsilon = \epsilon_1 + \epsilon_2 + \epsilon_3$ and let $c \in \mathbb{Z}^+$ be large enough so that
\[
0 \leq \displaystyle\lim_{r \rightarrow \infty}\rho_{r,c}(S:T) - \displaystyle\lim_{r \rightarrow \infty}\lim_{t \rightarrow \infty}\rho_{r,t}(S:T) \leq \epsilon_1
\]
and
\[
0 \leq \lim_{t \rightarrow \infty}\lim_{r \rightarrow \infty}\rho_{r,t}(S:T) - \lim_{t \rightarrow \infty}\rho_{c,t}(S:T) \leq \epsilon_2.
\]
Now, we let $d \in \mathbb{Z}^+$ be large enough so that
\begin{align}\label{0ep1}
0 \leq \rho_{d,c}(S:T) - \displaystyle\lim_{r \rightarrow \infty}\lim_{t \rightarrow \infty}\rho_{r,t}(S:T) \leq \epsilon_1,
\end{align}
\begin{align}\label{0ep2}
0 \leq \lim_{t \rightarrow \infty}\lim_{r \rightarrow \infty}\rho_{r,t}(S:T) - \rho_{c,d}(S:T) \leq \epsilon_2,
\end{align}
and
\begin{align}\label{o3}
e_c^k(d') \leq \epsilon_3,
\end{align}
where $e_c^k(d')$ comes from the eighth statement of Lemma \ref{mcrs}. By (\ref{0ep1}), (\ref{0ep2}), (\ref{o3}), and Lemma \ref{mcrs}, we have
\begin{align*}
mdim_{FS}(S:T) - \displaystyle\lim_{t \rightarrow \infty}\lim_{r \rightarrow \infty}\rho_{r,t}(S:T) &= \displaystyle\lim_{r \rightarrow \infty}\lim_{t \rightarrow \infty}\rho_{r,t}(S:T) - \displaystyle\lim_{t \rightarrow \infty}\lim_{r \rightarrow \infty}\rho_{r,t}(S:T)\\
								 &\geq \rho_{d,c}(S:T) - \rho_{c,d}(S:T) - \epsilon_1 - \epsilon_2\\
								 &\geq -e_c^k(d') - \epsilon_1 - \epsilon_2\\
								 &\geq -\epsilon_1 - \epsilon_2 - \epsilon_3\\
								 &= -\epsilon.
\end{align*}
Since the $\epsilon$ is arbitrary, we have
\[
mdim_{FS}(S:T) \geq \displaystyle\lim_{t \rightarrow \infty}\lim_{r \rightarrow \infty}\rho_{r,t}(S:T).
\]
As identical argument can be made to prove the equivalence of the iterated limits for $Mdim_{FS}(S:T)$. \qedhere
\end{proof}

The final theorem of this section describes the basic properties of finite-state mutual dimension.

\begin{thm}[Properties of Finite-State Mutual Dimensions]\label{md}
For all $S,T \in \Sigma^\infty$,
{\footnotesize
\begin{enumerate}
\item $dim_{FS}(S) {+} dim_{FS}(T) {-} Dim_{FS}(S,T) {\leq} mdim_{FS}(S:T) {\leq} Dim_{FS}(S) {+} Dim_{FS}(T) {-} Dim_{FS}(S,T)$,
\item $dim_{FS}(S) {+} dim_{FS}(T) {-} dim_{FS}(S,T) {\leq} Mdim_{FS}(S:T) {\leq} Dim_{FS}(S) {+} Dim_{FS}(T) {-} dim_{FS}(S,T)$,
\item $mdim_{FS}(S:T) \leq \min\{dim_{FS}(S), dim_{FS}(T)\}$, $Mdim_{FS}(S:T) \leq \min\{Dim_{FS}(S),Dim_{FS}(T)\}$,
\item $0 \leq mdim_{FS}(S:T) \leq Mdim_{FS}(S:T) \leq 1$,
\item $mdim_{FS}(S:S) = dim_{FS}(S)$, $Mdim_{FS}(S:S) = Dim_{FS}(S)$, and
\item $mdim_{FS}(S:T) = mdim_{FS}(T:S)$, $Mdim_{FS}(S:T) = Mdim_{FS}(T:S)$.
\end{enumerate}
}
\end{thm}

\begin{proof}
To prove the first statement, observe that, by the first statement of Lemma \ref{mcrs},
\begin{align*}
mdim_{FS}(S:T) &= \displaystyle\lim_{r \rightarrow \infty}\lim_{t \rightarrow \infty}\rho_{r,t}(S:T)\\
							 &\leq \displaystyle\lim_{r \rightarrow \infty}\lim_{t \rightarrow \infty}[\hat{\rho}_t(S) + \hat{\rho}_t(T) - \hat{\rho}_r(S,T)]\\
							 &= \lim_{t \rightarrow \infty}\hat{\rho}_t(S) + \lim_{t \rightarrow \infty}\hat{\rho}_t(T) - \displaystyle\lim_{r \rightarrow \infty}\hat{\rho}_r(S,T)\\
							 &= Dim_{FS}(S) + Dim_{FS}(T) - Dim_{FS}(S,T).
\end{align*}
Likewise,
\begin{align*}
mdim_{FS}(S:T) &= \displaystyle\lim_{r \rightarrow \infty}\lim_{t \rightarrow \infty}\rho_{r,t}(S:T)\\
							 &\geq \displaystyle\lim_{r \rightarrow \infty}\lim_{t \rightarrow \infty}[\rho_t(S) + \rho_t(T) - \hat{\rho}_r(S,T)]\\
							 &= \lim_{t \rightarrow \infty}\rho_t(S) + \lim_{t \rightarrow \infty}\rho_t(T) - \displaystyle\lim_{r \rightarrow \infty}\hat{\rho}_r(S,T)\\
							 &= dim_{FS}(S) + dim_{FS}(T) - Dim_{FS}(S,T).
\end{align*}
The proof of the second statement is similar to the proof of the first statement. To prove the third statement, we observe that, by the third statement of Lemma \ref{mcrs},
\begin{align*}
mdim_{FS}(S:T) &= \displaystyle\lim_{r \rightarrow \infty}\lim_{t \rightarrow \infty}\rho_{r,t}(S:T)\\
							 &\leq \lim_{r \rightarrow \infty}\lim_{t \rightarrow \infty}\bigg[\rho_t(S) + g_r^k(t')\bigg]\\
							 &= dim_{FS}(S).
\end{align*}
By a similar argument, we can prove that $mdim_{FS}(S:T) \leq dim_{FS}(T)$, and thus $mdim_{FS}(S:T) \leq \min\{dim_{FS}(S),\min\{dim_{FS}(T)\}$. Using a similar argument, we can also prove that $Mdim_{FS}(S:T) \leq \min\{Dim_{FS}(S),Dim_{FS}(T)\}$. To prove the fourth statement, observe that, by the fourth statement of Lemma \ref{mcrs} and Theorem \ref{interchange}, 
\begin{align*}
mdim_{FS}(S:T) &= \displaystyle\lim_{t \rightarrow \infty}\lim_{r \rightarrow \infty}\rho_{r,t}(S:T)\\
							 &\geq \displaystyle\lim_{t \rightarrow \infty}\lim_{r \rightarrow \infty}-h_t^k(r')\\
							 &= 0.
\end{align*}
Now, since the upper finite-state dimension of a sequence is no larger than one, we have
\begin{align*}
mdim_{FS}(S:T) &\leq Mdim_{FS}(S:T)\\
							 &\leq \min\{Dim_{FS}(S),Dim_{FS}(T)\}\\
							 &\leq 1.
\end{align*}
To prove the fifth statement, observe that, by the fifth statement of Lemma \ref{mcrs} and Theorem \ref{interchange},
\begin{align*}
mdim_{FS}(S:S) &= \displaystyle\lim_{t \rightarrow \infty}\lim_{r \rightarrow \infty}\rho_{r,t}(S:S)\\
							 &\geq \displaystyle\lim_{t \rightarrow \infty}\lim_{r \rightarrow \infty}\bigg[\rho_t(S) - i_t^k(r')\bigg]\\
							 &= \displaystyle\lim_{t \rightarrow \infty}\rho_t(S)\\
							 &= dim_{FS}(S).
\end{align*}
Also, by the sixth statement of Lemma \ref{mcrs},
\begin{align*}
mdim_{FS}(S:S) &= \displaystyle\lim_{r \rightarrow \infty}\lim_{t \rightarrow \infty}\rho_{r,t}(S:S)\\
							 &\leq \displaystyle\lim_{r \rightarrow \infty}\lim_{t \rightarrow \infty}\bigg[\rho_t(S) + j_r^k(t')\bigg]\\
							 &= \displaystyle\lim_{t \rightarrow \infty}\rho_t(S)\\
							 &= dim_{FS}(S).
\end{align*}
Using a similar argument, we can prove that $Mdim_{FS}(S:S) = Dim_{FS}(S)$. Finally, to prove the sixth statement, observe that, by the seventh statement of Lemma \ref{mcrs},
\begin{align*}
mdim_{FS}(S:T) &= \displaystyle\lim_{r \rightarrow \infty}\lim_{t \rightarrow \infty}\rho_{r,t}(S:T)\\
							 &= \displaystyle\lim_{r \rightarrow \infty}\lim_{t \rightarrow \infty}\rho_{r,t}(T:S)\\
							 &= mdim_{FS}(T:S).
\end{align*}
By a similar argument, we can show that $Mdim_{FS}(S:T) = Mdim_{FS}(T:S)$. \qedhere
\end{proof}


\section{Block Mutual Information Rates}
In this section, we introduce the notion of \emph{block mutual information rates} between sequences and prove that the lower and upper finite-state mutual dimensions can be characterized in terms of block mutual information rates.

Originally, Ziv and Lempel proved that the upper finite-state compression ratio of a sequence may be characterized in terms of the \emph{entropy rates} of \emph{non-aligned} block frequencies \cite{jZivLem78} within the sequence. Sheinwald proved a similar characterization of the upper compression ratio using the \emph{entropy rates} of \emph{aligned} block frequencies \cite{jShei94}. Later, Bourke, Hitchcock, and Vindochandran proved a characterization of the lower and upper finite-state dimensions of sequences \cite{jBoHiVi05} in terms of (aligned) block entropy rates. Kozachinskiy and Shen recently proved that the lower finite-state dimension can also be characterized using the entropy rates of non-aligned block frequencies \cite{cKozShe2019}.

We begin by discussing Shannon $\emph{mutual information}$.

\begin{defn}
Let $\alpha$ be a discrete probability measure on $\mathcal{X} \times \mathcal{X}$. The \emph{Shannon mutual information} between $\alpha_1$ and $\alpha_2$ is
\[
\mi(\alpha_1;\alpha_2) = \ent(\alpha_1) + \ent(\alpha_2) - \ent(\alpha).
\]
\end{defn}

By the properties of Shannon entropy found in Theorem \ref{Sha}, we have the following properties regarding mutual information.

\begin{thm}\label{m}
Let $\alpha$ be a probability measure on $\mathcal{X} \times \mathcal{X}$.
\begin{enumerate}
\item $\emph{\mi}(\alpha_1;\alpha_2) \geq 0$.
\item $\emph{\mi}(\alpha_1;\alpha_2) \leq \min\{\emph{\ent}(\alpha_1),\emph{\ent}(\alpha_2)\}$.
\item If $\sum_{a \in \Sigma}\alpha(a,a) = 1$, then $\emph{\mi}(\alpha_1;\alpha_2) = \emph{\ent}(\alpha_1) = \emph{\ent}(\alpha_2) = \emph{\ent}(\alpha)$.
\item $\emph{\mi}(\alpha_1;\alpha_2) = \emph{\mi}(\alpha_2;\alpha_1)$.
\end{enumerate}
\end{thm}

Let $\ell,n \in \mathbb{Z}^+$ such that $n$ is a multiple of $\ell$ and let $u,w \in \Sigma^n$. By applying Theorem \ref{m} to $\pi_{u,w}^{(\ell)}$, where $\alpha = \pi_{u,w}^{(\ell)}$, $\alpha_1 = \pi_u^{(\ell)}$, and $\alpha_2 = \pi_w^{(\ell)}$, we obtain the following corollary.

\begin{cor}\label{pmi}
For every $n,\ell \in \mathbb{Z}^+$ such that $n$ is a multiple of $\ell$ and all $u \in \Sigma^n$ and $w \in \Sigma^n$,
\begin{enumerate}
\item $\emph{\mi}(\pi_u^{(\ell)};\pi_w^{(\ell)}) \geq 0$,
\item $\emph{\mi}(\pi_u^{(\ell)};\pi_w^{(\ell)}) \leq \min\{\emph{\ent}(\pi_u^{(\ell)}),\emph{\ent}(\pi_w^{(\ell)})\}$,
\item $\emph{\mi}(\pi_u^{(\ell)};\pi_u^{(\ell)}) = \emph{\ent}(\pi_u^{(\ell)})$, and
\item $\emph{\mi}(\pi_u^{(\ell)};\pi_w^{(\ell)}) = \emph{\mi}(\pi_w^{(\ell)};\pi_u^{(\ell)})$.
\end{enumerate}
\end{cor}

We now proceed to prove several lemmas which provide bounds on the difference of the normalized mutual information between the block frequencies of two strings and the mutual compression ratio between the same two strings. These lemmas will be needed to prove the main theorem of this section.

\begin{lem}\label{mitomc}
For all $r,t \in \mathbb{Z}^+$ and $u,w \in \Sigma^n$ such that $n \geq r'$,
\[
\frac{\emph{\mi}(\pi_{u_{r'}}^{(r')};\pi_{w_{r'}}^{(r')})}{r' \log k} - \rho_{r,t}(u:w) \leq 2\Big\lfloor \frac{n}{r'} \Big\rfloor^{-1} + q_t^k(r'),
\]
where $r' = \lfloor \log_k r \rfloor$, $u_{r'} = u \harp \Big \lfloor \frac{n}{r'}  \Big \rfloor \cdot r'$, $w_{r'} = w \harp \Big \lfloor \frac{n}{r'}  \Big \rfloor \cdot r'$, and $\displaystyle\lim_{m \rightarrow \infty}q_t^k(m) = 0$.
\end{lem}

\begin{proof}
By Lemmas \ref{hc} and \ref{huffc},
\begin{align*}
&\frac{\mi(\pi_{u_{r'}}^{(r')};\pi_{w_{r'}}^{(r')})}{r' \log k} - \rho_{r,t}(u:w)\\
&= \frac{\ent(\pi_{u_{r'}}^{(r')})}{r' \log k} + \frac{\ent(\pi_{w_{r'}}^{(r')})}{r' \log k} - \frac{\ent(\pi_{u_{r'},w_{r'}}^{(r')})}{r' \log k} - \rho_t(u) - \rho_t(w) + \rho_r(u,w)\\
&= \frac{\ent(\pi_{u_{r'}}^{(r')})}{r' \log k} + \frac{\ent(\pi_{w_{r'}}^{(r')})}{r' \log k} - \frac{\ent(\pi_{u_{r'},w_{r'}}^{(r')})}{r' \log k} - \rho_t(u) - \rho_t(w) + 2\rho_r((u,w))\\
&\leq \rho_t(u) + \rho_t(w) - \frac{\ent(\pi_{u_{r'},w_{r'}}^{(r')})}{r' \log k} - \rho_t(u) - \rho_t(w) + 2\bigg(\frac{\ent\big(\pi_{(u_{r'},w_{r'})}^{(r')}\big)}{r' \log k^2} + \frac{1}{r'}\bigg) + 2\Big\lfloor \frac{n}{r'} \Big\rfloor^{-1} + 2f_t^k(r')\\
&= \rho_t(u) + \rho_t(w) - \frac{\ent(\pi_{u_{r'},w_{r'}}^{(r')})}{r' \log k} - \rho_t(u) - \rho_t(w) + \frac{\ent(\pi_{u_{r'},w_{r'}}^{(r')})}{r' \log k} + \frac{2}{r'} + 2\Big\lfloor \frac{n}{r'} \Big\rfloor^{-1} + 2f_t^k(r')\\
&= \frac{2}{r'} + 2\Big\lfloor \frac{n}{r'} \Big\rfloor^{-1} + 2f_t^k(r')\\
&= 2\Big\lfloor \frac{n}{r'} \Big\rfloor^{-1} + q_t^k(r'),
\end{align*}
where $q_t^k(r') = \frac{2}{r'} + 2f_t^k(r')$. \qedhere
\end{proof}

\begin{lem}\label{mctomi}
For all $r,t \in \mathbb{Z}^+$ and $u,w \in \Sigma^n$ such that $n \geq r'$,
\[
\rho_{t,r}(u:w) - \frac{\emph{\mi}(\pi_{u_{r'}}^{(r')};\pi_{w_{r'}}^{(r')})}{r'\log k} \leq \Big\lfloor \frac{n}{r'} \Big\rfloor^{-1} + p_t^k(r'),
\]
where $r' = \lfloor \log_k r \rfloor$, $u_{r'} = u \harp \Big \lfloor \frac{n}{r'}  \Big \rfloor \cdot r'$, $w_{r'} = w \harp \Big \lfloor \frac{n}{r'}  \Big \rfloor \cdot r'$, and $\displaystyle\lim_{m \rightarrow \infty}p_t^k(m) = 0$.
\end{lem}

\begin{proof}
By Lemmas \ref{hc2} and \ref{huffc},
\begin{align*}
&\rho_{t,r}(u:w) - \frac{\mi(\pi_{u_{r'}}^{(r')};\pi_{w_{r'}}^{(r')})}{r'\log k}\\
&= \rho_r(u) + \rho_r(w) - \rho_t(u,w) - \frac{\ent(\pi_{u_{r'}}^{(r')})}{r'\log k} - \frac{\ent(\pi_{w_{r'}}^{(r')})}{r'\log k} + \frac{\ent(\pi_{u_{r'},w_{r'}}^{(r')})}{r'\log k}\\
&\leq \frac{\ent(\pi_{u_{r'}}^{(r')})}{r'\log k} + \frac{\ent(\pi_{w_{r'}}^{(r')})}{r'\log k} - \frac{\ent(\pi_{u_{r'},w_{r'}}^{(r')})}{r'\log k} - \frac{\ent(\pi_{u_{r'}}^{(r')})}{r'\log k} - \frac{\ent(\pi_{w_{r'}}^{(r')})}{r'\log k} + \frac{\ent(\pi_{u_{r'},w_{r'}}^{(r')})}{r'\log k} + \frac{2}{r'} + \Big\lfloor \frac{n}{r'} \Big\rfloor^{-1} + f_t^{k^2}(r') \\
&= \frac{2}{r'} + \Big\lfloor \frac{n}{r'} \Big\rfloor^{-1} + f_t^{k^2}(r')\\
&= \Big\lfloor \frac{n}{r'} \Big\rfloor^{-1} + p_t^k(r'),
\end{align*}
where $p_t^k(r') = \frac{2}{r'} + f_t^{k^2}(r')$. \qedhere
\end{proof}

We now discuss the $\ell^{th}$ block entropy rates of sequences. For any $n,m \in \mathbb{Z}^+$, $x,y \in \Sigma^m$, and $S,T \in \Sigma^{\infty}$, we denote the $n^{th}$ \emph{block frequency} of $x$ in $S$ by the function $\pi_{S,n}: \Sigma^* \rightarrow \mathbb{Q}_{[0,1]}$, defined by
\[
\pi_{S,n}(x) = \pi_{S \harp nm}(x) = \frac{\#_\Box(x,S \harp nm)}{n},
\]
and the $n^{th}$ \emph{joint block frequency} of $x$ in $S$ and $y$ in $T$ by the function $\pi_{S,T,n}: \Sigma^* \times \Sigma^* \rightarrow \mathbb{Q}_{[0,1]}$, defined by
\[
\pi_{S,T,n}(x,y) = \pi_{S \harp nm,T \harp nm}(x) = \frac{\#_\Box((x,y),(S,T) \harp nm)}{n}.
\]
As before, for each $\ell \in \mathbb{Z}^+$, we denote the restriction of $\pi_{S,n}$ to the strings in $\Sigma^\ell$ by $\pi_{S,n}^{(\ell)}$ and the restriction of $\pi_{S,T,n}$ to the pairs of strings in $\Sigma^\ell \times \Sigma^\ell$ by $\pi_{S,T,n}^{(\ell)}$.

\begin{defn}
Let $\ell \in \mathbb{Z}^+$ and $S \in \Sigma^\infty$. The $\ell^{th}$ \emph{lower } and \emph{upper block entropy rates of} $S$ are
\[
\ent_{\ell}(S) = \frac{1}{\ell \log k}\displaystyle\liminf_{n \rightarrow \infty}\ent(\pi_{S,n}^{(\ell)})
\]
and
\[
\hat{\ent}_{\ell}(S) = \frac{1}{\ell \log k}\displaystyle\limsup_{n \rightarrow \infty}\ent(\pi_{S,n}^{(\ell)})
\]
respectively.
\end{defn}

\begin{defn}
Let $\ell \in \mathbb{Z}^+$ and $S \in \Sigma^\infty$. The $\ell^{th}$ \emph{lower} and \emph{upper joint block entropy rates of} $S$ and $T$ are
\[
\ent_{\ell}(S,T) = \frac{1}{\ell \log k}\displaystyle\liminf_{n \rightarrow \infty}\ent(\pi_{S,T,n}^{(\ell)})
\]
and
\[
\hat{\ent}_{\ell}(S,T) = \frac{1}{\ell \log k}\displaystyle\limsup_{n \rightarrow \infty}\ent(\pi_{S,T,n}^{(\ell)})
\]
respectively.
\end{defn}

We make note that the $\ell^{th}$ lower and upper block entropy rates $\ent_{\ell}((S,T))$ and $\hat{\ent}_{\ell}((S,T))$ of $(S,T) \in (\Sigma \times \Sigma)^\infty$ are normalized by $\ell \log k^2$ and the $\ell^{th}$ lower and upper joint block entropy rates $\ent_{\ell}(S,T)$ and $\hat{\ent}_{\ell}(S,T)$ of $S \in \Sigma^\infty$ and $T \in \Sigma^\infty$ are normalized by $\ell \log k$.

\begin{defn}
Let $\ell \in \mathbb{Z}^+$ and $S,T \in \Sigma^{\infty}$. The $\ell^{th}$ \emph{lower} and \emph{upper block mutual information rates between} $S$ and $T$ are
\[
\mi_{\ell}(S;T) = \frac{1}{\ell \log k}\displaystyle\liminf_{n \rightarrow \infty}\mi(\pi_{S,n}^{(\ell)};\pi_{T,n}^{(\ell)})
\]
and
\[
\hat{\mi}_{\ell}(S;T) = \frac{1}{\ell \log k}\displaystyle\limsup_{n \rightarrow \infty}\mi(\pi_{S,n}^{(\ell)};\pi_{T,n}^{(\ell)})
\]
respectively.
\end{defn}

The following theorem regarding the properties of block mutual information rates between sequences follows directly from Corollary \ref{pmi} and the definitions of $\liminf$ and $\limsup$.

\begin{lem}[Properties of $\ell^{th}$ Block Mutual Information Rates between Sequences]\label{bmur}
Let $\ell \in \mathbb{Z}^+$ and $S,T \in \Sigma^\infty$.
\begin{enumerate}
\item $\emph{\mi}_{\ell}(S;T) \geq 0$, $\hat{\emph{\mi}}_{\ell}(S;T) \geq 0$.
\item $\emph{\ent}_{\ell}(S) + \emph{\ent}_{\ell}(T) - \hat{\emph{\ent}}_{\ell}(S,T) \leq \emph{\mi}_{\ell}(S;T) \leq \hat{\emph{\ent}}_{\ell}(S) + \hat{\emph{\ent}}_{\ell}(T) - \hat{\emph{\ent}}_{\ell}(S,T)$.
\item $\emph{\ent}_{\ell}(S) + \emph{\ent}_{\ell}(T) - \emph{\ent}_{\ell}(S,T) \leq \hat{\emph{\mi}}_{\ell}(S;T) \leq \hat{\emph{\ent}}_{\ell}(S) + \hat{\emph{\ent}}_{\ell}(T) - \emph{\ent}_{\ell}(S,T)$.
\item $\emph{\mi}_{\ell}(S;T) \leq \min\{\emph{\ent}_{\ell}(S),\emph{\ent}_{\ell}(T)\}$, $\hat{\emph{\mi}}_{\ell}(S;T) \leq \min\{\hat{\emph{\ent}}_{\ell}(S),\hat{\emph{\ent}}_{\ell}(T)\}$.
\item $\emph{\mi}_{\ell}(S;S) = \emph{\ent}_{\ell}(S)$, $\hat{\emph{\mi}}_{\ell}(S;S) = \hat{\emph{\ent}}_{\ell}(S)$.
\item $\emph{\mi}_{\ell}(S;T) = \emph{\mi}_{\ell}(T;S)$, $\hat{\emph{\mi}}_{\ell}(S;T) = \hat{\emph{\mi}}_{\ell}(T;S)$.
\end{enumerate}
\end{lem}

We now make an observation that will be used to prove two lemmas that provide upper-bounds on the difference of the block mutual information and mutual compression ratio between two sequences.

\begin{obs}\label{obs3}
For any $\ell \in \mathbb{Z}^+$ and $S,T \in \Sigma^\infty$,
\[
\mi_{\ell}(S;T) = \frac{1}{\ell \log k}\displaystyle\liminf_{n \rightarrow \infty}\mi(\pi_{S,\lfloor\frac{n}{\ell}\rfloor}^{(\ell)};\pi_{T,\lfloor\frac{n}{\ell}\rfloor}^{(\ell)})
\]
and
\[
\hat{\mi}_{\ell}(S;T) = \frac{1}{\ell \log k}\displaystyle\limsup_{n \rightarrow \infty}\mi(\pi_{S,\lfloor\frac{n}{\ell}\rfloor}^{(\ell)};\pi_{T,\lfloor\frac{n}{\ell}\rfloor}^{(\ell)}).
\]
\end{obs}

\begin{lem}\label{mitomc2}
For all $r,t \in \mathbb{Z}^+$ and $S,T \in \Sigma^\infty$,
\[
\emph{\mi}_{r'}(S;T) - \rho_{r,t}(S:T) \leq q_t^k(r')
\]
and
\[
\hat{\emph{\mi}}_{r'}(S;T) - \hat{\rho}_{r,t}(S:T) \leq q_t^k(r')
\]
where $r' = \lfloor \log_k r \rfloor$ and $\displaystyle\lim_{m \rightarrow \infty}q_t^k(m) = 0$.
\end{lem}

\begin{proof}
By Lemma \ref{mitomc} and Observation \ref{obs3},
\begin{align*}
&\mi_{r'}(S;T) - \rho_{r,t}(S:T)\\
&=\frac{1}{r' \log k}\displaystyle\liminf_{n \rightarrow \infty}\mi(\pi_{S,\lfloor\frac{n}{r'}\rfloor}^{(r')};\pi_{T,\lfloor\frac{n}{r'}\rfloor}^{(r')}) -\displaystyle\liminf_{n \rightarrow \infty}\rho_{r,t}(S \harp n:T \harp n)\\
&=\displaystyle\liminf_{n \rightarrow \infty}\frac{\mi(\pi_{S,\lfloor\frac{n}{r'}\rfloor}^{(r')};\pi_{T,\lfloor\frac{n}{r'}\rfloor}^{(r')})}{r' \log k} -\displaystyle\liminf_{n \rightarrow \infty}\rho_{r,t}(S \harp n:T \harp n)\\
&\leq \displaystyle\limsup_{n \rightarrow \infty}\bigg[\frac{\mi(\pi_{S,\lfloor\frac{n}{r'}\rfloor}^{(r')};\pi_{T,\lfloor\frac{n}{r'}\rfloor}^{(r')})}{r' \log k} -\rho_{r,t}(S \harp n:T \harp n)\bigg]\\
&\leq \displaystyle\limsup_{n \rightarrow \infty}\bigg[ 2\Big\lfloor \frac{n}{r'} \Big\rfloor^{-1} + q_t^k(r') \bigg]\\
&= q_t^k(r').
\end{align*}
An identical argument can be given to prove the second inequality. \qedhere
\end{proof}

\begin{lem}\label{mctomi2}
For all $r,t \in \mathbb{Z}^+$ and $S,T \in \Sigma^\infty$,
\[
\rho_{t,r}(S:T) - \emph{\mi}_{r'}(S;T) \leq p_t^k(r')
\]
and
\[
\hat{\rho}_{t,r}(S:T) - \hat{\emph{\mi}}_{r'}(S;T) \leq p_t^k(r')
\]
where $r' = \lfloor \log_k r \rfloor$ and $\displaystyle\lim_{m \rightarrow \infty}p_t^k(m) = 0$.
\end{lem}

\begin{proof}
By Lemma \ref{mctomi} and Observation \ref{obs3},
\begin{align*}
&\rho_{t,r}(S:T) - \mi_{r'}(S;T)\\
&= \displaystyle\liminf_{n \rightarrow \infty}\rho_{t,r}(S \harp n:T \harp n) - \frac{1}{r'\log k}\displaystyle\liminf_{n \rightarrow \infty}\mi(\pi_{S,\lfloor\frac{n}{r'}\rfloor}^{(r')};\pi_{T,\lfloor\frac{n}{r'}\rfloor}^{(r')})\\
&= \displaystyle\liminf_{n \rightarrow \infty}\rho_{t,r}(S \harp n:T \harp n) - \displaystyle\liminf_{n \rightarrow \infty}\frac{\mi(\pi_{S,\lfloor\frac{n}{r'}\rfloor}^{(r')};\pi_{T,\lfloor\frac{n}{r'}\rfloor}^{(r')})}{r'\log k}\\
&\leq \displaystyle\limsup_{n \rightarrow \infty}\bigg[\rho_{t,r}(S \harp n:T \harp n) - \frac{\mi(\pi_{S,\lfloor\frac{n}{r'}\rfloor}^{(r')};\pi_{T,\lfloor\frac{n}{r'}\rfloor}^{(r')})}{r'\log k} \bigg]\\
&\leq \displaystyle\limsup_{n \rightarrow \infty}\bigg[ \Big\lfloor \frac{n}{r'} \Big\rfloor^{-1} + p_t^k(r') \bigg]\\
&= p_t^k(r').
\end{align*}
An identical argument can be given to prove the second inequality. \qedhere
\end{proof}

We now discuss the block entropy rates and joint block entropy rates of sequences and introduce block mutual information rates between two sequences.

\begin{defn}
The \emph{lower} and \emph{upper block entropy rates of} $S$ are
\[
\ent(S) = \displaystyle\lim_{\ell \rightarrow \infty}\ent_{\ell}(S)
\]
and
\[
\hat{\ent}(S) = \displaystyle\lim_{\ell \rightarrow \infty}\hat{\ent}_{\ell}(S),
\]
respectively.
\end{defn}

\begin{defn}
The \emph{lower} and \emph{upper joint block entropy rates of} $S \in \Sigma^\infty$ and $T \in \Sigma^\infty$ are
\[
\ent(S,T) = \displaystyle\lim_{\ell \rightarrow \infty}\ent_{\ell}(S,T)
\]
and
\[
\hat{\ent}(S,T) = \displaystyle\lim_{\ell \rightarrow \infty}\hat{\ent}_{\ell}(S,T),
\]
respectively.
\end{defn}

Using the frameworks developed in \cite{jZivLem78} and \cite{jDaLaLuMa04}, Bourke, Hitchcock, and Vinodchandran proved the following theorem in \cite{jBoHiVi05}.

\begin{thm}[\cite{jBoHiVi05}]\label{dimenteq}
For every $S \in \Sigma^\infty$,
\[
dim_{FS}(S) = \emph{\ent}(S)
\]
and
\[
Dim_{FS}(S) = \hat{\emph{\ent}}(S).
\]
\end{thm}

The following corollary follows directly from Theorem \ref{dimenteq}.

\begin{cor}\label{dimenteqcor}
For every $S,T \in \Sigma^\infty$,
\[
dim_{FS}(S,T) = \emph{\ent}(S,T)
\]
and
\[
Dim_{FS}(S,T) = \hat{\emph{\ent}}(S,T).
\]
\end{cor}

\begin{defn}
The \emph{lower} and \emph{upper block mutual information rates between} $S \in \Sigma^\infty$ and $T \in \Sigma^\infty$ are
\begin{align}\label{lmir}
\mi(S;T) = \displaystyle\lim_{\ell \rightarrow \infty}\mi_{\ell}(S;T)
\end{align}
and
\begin{align}\label{umir}
\hat{\mi}(S;T) = \displaystyle\lim_{\ell \rightarrow \infty}\hat{\mi}_{\ell}(S;T),
\end{align}
respectively.
\end{defn}

We now present the main theorem of this section, which states that the lower and upper block mutual information rates coincide with the lower and upper finite-state mutual dimensions, respectively.

\begin{thm}\label{char}
For all $S,T \in \Sigma^\infty$,
\[
mdim_{FS}(S:T) = \emph{\mi}(S;T)
\]
and
\[
Mdim_{FS}(S:T) = \hat{\emph{\mi}}(S;T).
\]
\end{thm}

\begin{proof}
Let $\epsilon > 0$, $\epsilon_1 > 0$, and $\epsilon_2 > 0$ such that $\epsilon = \epsilon_1 + \epsilon_2$. First, let $c \in \mathbb{Z}^+$ such that,
\begin{align}\label{limeq1}
\big|\displaystyle\lim_{r \rightarrow \infty}\rho_{r,c}(S:T) - mdim_{FS}(S:T) \big| < \epsilon_1
\end{align}
and
\begin{align}\label{limeq2}
\big|\displaystyle\lim_{t \rightarrow \infty}\rho_{c,t}(S:T) - mdim_{FS}(S:T) \big| < \epsilon_1.
\end{align}
We then choose a constant $D \in \mathbb{Z}^+$ such that, for all $d > D$,
\begin{align}\label{closer}
q_c^k(d') < \epsilon_2 \text{\,\,\,\, and \,\,\,\,} p_c^k(d') < \epsilon_2
\end{align}
where $d' = \lfloor \log_k d \rfloor$ and $q_c^k(d')$ and $p_c^k(d')$ are from Lemmas \ref{mitomc2} and \ref{mctomi2}, respectively. By Lemma \ref{mitomc2}, (\ref{limeq1}), and (\ref{closer}),
\begin{align*}
\mi_{d'}(S;T) - mdim_{FS}(S:T) &\leq \rho_{d,c}(S:T) - mdim_{FS}(S:T) + q_c^k(d')\\
															 &\leq \displaystyle\lim_{r \rightarrow \infty}\rho_{r,c}(S:T) - mdim_{FS}(S:T) + q_c^k(d')\\
															 &\leq \epsilon_1 + \epsilon_2\\
															 &= \epsilon.
\end{align*}
Likewise, by Lemma \ref{mctomi2}, (\ref{limeq2}), and (\ref{closer}),
\begin{align*}
\mi_{d'}(S;T) - mdim_{FS}(S:T) &\geq \rho_{c,d}(S:T) - mdim_{FS}(S:T) - p_c^k(d')\\
														 &\geq \displaystyle\lim_{t \rightarrow \infty}\rho_{c,t}(S:T) - mdim_{FS}(S:T) - p_c^k(d')\\
														 &\geq -\epsilon_1 - \epsilon_2\\
														 &= -\epsilon.
\end{align*}
Therefore, for every $\epsilon > 0$, there exists a constant $D \in \mathbb{Z}^+$ such that, for all $d > D$,
\[
|\mi_{d'}(S;T) - mdim_{FS}(S:T)| < \epsilon,
\]
which proves that the limit from the definition of $\mi(S;T)$ exists and is equal to $mdim_{FS}(S:T)$. An identical argument can be given to prove that the limit from the definition of $\hat{\mi}(S;T)$ exists and is equal to $Mdim_{FS}(S:T)$. \qedhere
\end{proof}

The following theorem regarding the properties of block-mutual information rates between sequences follows from Theorem \ref{md}, Theorem \ref{dimenteq}, Corollary \ref{dimenteqcor}, and Theorem \ref{char}. This theorem may also be proven using the properties listed in Lemma \ref{bmur}.

\begin{thm}[Properties of Block Mutual Information Rates between Sequences]\label{mirp}
For all $S,T \in \Sigma^\infty$,
\begin{enumerate}
\item $\emph{\ent}(S) + \emph{\ent}(T) - \hat{\emph{\ent}}(S,T) \leq \emph{\mi}(S;T) \leq \hat{\emph{\ent}}(S) + \hat{\emph{\ent}}(T) - \hat{\emph{\ent}}(S,T)$
\item $\emph{\ent}(S) + \emph{\ent}(T) - \emph{\ent}(S,T) \leq \hat{\emph{\mi}}(S;T) \leq \hat{\emph{\ent}}(S) + \hat{\emph{\ent}}(T) - \emph{\ent}(S,T)$
\item $\emph{\mi}(S;T) \leq \min\{\emph{\ent}(S),\emph{\ent}(T)\}$, $\hat{\emph{\mi}}(S;T) \leq \min\{\hat{\emph{\ent}}(S),\hat{\emph{\ent}}(T)\}$,
\item $0 \leq \emph{\mi}(S;T) \leq \hat{\emph{\mi}}(S;T) \leq 1$,
\item $\emph{\mi}(S;S) = \emph{\ent}(S)$, $\hat{\emph{\mi}}(S;S) = \hat{\emph{\ent}}(S)$, and
\item $\emph{\mi}(S;T) = \emph{\mi}(T;S)$, $\hat{\emph{\mi}}(S;T) = \hat{\emph{\mi}}(T;S)$.
\end{enumerate}
\end{thm}


\section{Finite-State Mutual Dimension and Independence}
In this section we explore some of the relationships between finite-state mutual dimension and normal sequences. More specifically, we provide necessary and sufficient conditions for when two normal sequences achieve finite-state mutual dimension zero.

Becher, Carton, and Heiber provided a notion of \emph{finite-state independence} using the \emph{conditional compression ratio} of a sequence \emph{given} another sequence. Specifically, they define two sequences $S \in \Sigma^{\infty}$ and $T \in \Sigma^{\infty}$ to be \emph{finite-state independent} if the conditional compression ratio $\rho(S\,|\,T)$ of $S$ given $T$ is equal to the compression ratio $\rho(S)$ of $S$, the conditional compression ratio $\rho(T\,|\,S)$ of $T$ given $S$ is equal to the compression ratio $\rho(T)$ of $T$, and both $\rho(S)$ and $\rho(T)$ are greater than zero. In their investigation they showed that, for any two normal sequences $R_1 \in \Sigma^{\infty}$ and $R_2 \in \Sigma^{\infty}$, if $R_1$ and $R_2$ are finite-state independent, then $(R_1,R_2)$ is normal. However, they also showed that the converse does not hold, i.e., there are two normal sequences $R_1$ and $R_2$ such that $(R_1,R_2)$ is normal and not finite-state independent \cite{jBeCaHe18}. Alvarez, Becher, and Carton also proved several characterizations of finite-state independence using various kinds of B\"uchi automata \cite{jAlBeCa19}.

We now proceed to discuss the concept of normality and its relationship to finite-state dimension.

\begin{defn}
Let $\alpha$ be a probability measure on $\Sigma$, $S \in \Sigma^\infty$, and $\ell \in \mathbb{Z}^+$.
\begin{enumerate}
\item $S$ is $\alpha$-$\ell$-\emph{normal} if, for all $x \in \Sigma^\ell$,
\[
\displaystyle\lim_{n \rightarrow \infty}\pi_{S,n}(x) = \alpha(x).
\]
\item $S$ is $\alpha$-\emph{normal} if $S$ is $\alpha$-$\ell$-normal for all $\ell \in \mathbb{Z}^+$.
\item $S$ is \emph{normal} if $S$ is $\mu$-normal, where $\mu$ is the uniform probability measure on $\Sigma$.
\item $S$ has \emph{asymptotic frequency} $\alpha$, and we write $S \in FREQ^{\alpha}$, if $S$ is $\alpha$-1-normal.
\end{enumerate}
\end{defn}
In \cite{jLutz11}, Lutz explored the \emph{lower} and \emph{upper finite-state} $\beta$-\emph{dimensions} $dim_{FS}^{\beta}(S)$ and $Dim_{FS}^{\beta}(S)$ of a sequence $S \in \Sigma^{\infty}$, where $\beta$ is a probability measure on $\Sigma$. These quantities are essentially finite-state versions of \emph{Billingsley dimension} and \emph{strong Billingsley dimension}, respectively \cite{jBill60}. We will need to use these concepts to prove our main theorem.

Let $\beta$ be a probability measure on $\Sigma$. The \emph{Shannon self-information} of a string $w \in \Sigma^*$ with respect to $\beta$ on $\Sigma$ is
\[
\ell_{\beta}(w) = \displaystyle\sum_{i=0}^{|w|-1}\log\frac{1}{\beta(w[i])}.
\]
The $\beta$-\emph{compression ratio of} $u \in \Sigma^*$ attained by an ILFSC $C$ on $\Sigma$ is
\[
\rho^\beta_C(u) = \frac{|C(u)|}{\ell_{\beta}(u)}.
\]

\begin{defn}
Let $\beta$ be a probability measure on $\Sigma$. The $r$-\emph{state} $\beta$-\emph{compression ratio of} $u \in \Sigma^*$ is
\[
\rho^\beta_r(u) = \min\big\{\rho^\beta_C(u) \,\bigg|\, C \text{ is an ILFSC on } \Sigma \text{ that has } r \text{ states} \big\}
\]
\end{defn}

\begin{defn}
Let $\beta$ be a probability measure on $\Sigma$. The \emph{lower} and \emph{upper finite-state} $\beta$-\emph{dimensions of} $S \in \Sigma^\infty$ are
\[
dim^\beta_{FS}(S) = \displaystyle\lim_{r \rightarrow \infty}\displaystyle\liminf_{n \rightarrow \infty}\rho^\beta_r(S \harp n)
\]
and
\[
Dim^\beta_{FS}(S) = \displaystyle\lim_{r \rightarrow \infty}\displaystyle\limsup_{n \rightarrow \infty}\rho^\beta_r(S \harp n),
\]
respectively.
\end{defn}

Schnorr and Stimm proved a characterization of normal sequences in terms of finite-state gamblers \cite{jSchSti72}. Later, Dai, Lathrop, Lutz, and Mayordomo showed that any normal sequence achieves finite-state dimension one \cite{jDaLaLuMa04}, while Bourke, Hitchcock, and Vinodchandran showed that any sequence that achieves finite-state dimension one is normal. \cite{jBoHiVi05}. This result can easily be generalized to $\alpha$-normal sequences.

\begin{thm}[\cite{jSchSti72,jBoHiVi05}]\label{normalchar}
For each probability measure $\alpha$ on $\Sigma$ and each $R \in \Sigma^{\infty}$, $R$ is $\alpha$-normal if and only if $dim^{\alpha}_{FS}(R) = 1$.
\end{thm}

\noindent The main theorem of this section provides a similar characterization for pairs of normal sequences that achieve finite-state mutual dimension zero.
\begin{thm}\label{independent}
Let $\alpha_1$ and $\alpha_2$ be positive probability measures on $\Sigma$. If $R_1$ is $\alpha_1$-normal and $R_2$ is $\alpha_2$-normal, then $(R_1,R_2)$ is $(\alpha_1 \times \alpha_2)$-normal if and only if $Mdim_{FS}(R_1:R_2) = 0$.
\end{thm}
\noindent Note that, in the above theorem, the \emph{product probability measure} $(\alpha_1 \times \alpha_2)$ on $\Sigma \times \Sigma$ is defined by
\[
(\alpha_1 \times \alpha_2)(a,b) = \alpha_1(a)\alpha_2(b),
\]
for all $a,b \in \Sigma$. We present the proof of Theorem \ref{independent} at the end of this section. Thus finite-state mutual dimension provides a mechanism in which to reason about the degree to which two sequences are independent of one another at the finite-state level.

First, we make the following observation regarding $\alpha$-normal sequences over $\Sigma \times \Sigma$.

\begin{obs}\label{obs4}
Let $\alpha$ be a probability measure on $\Sigma \times \Sigma$. If a sequence $(R_1,R_2) \in (\Sigma \times \Sigma)^\infty$ is $\alpha$-normal, then $R_1$ is $\alpha_1$-normal and $R_2$ is $\alpha_2$-normal.
\end{obs}

\noindent Lutz proved the following theorem about $\alpha$-normal sequences \cite{jLutz11}.

\begin{thm}[\cite{jLutz11}]\label{norm}
If $\alpha$ is a probability measure on $\Sigma$, then, for every $\alpha$-normal sequence $R \in \Sigma^\infty$,
\[
dim_{FS}(R) = Dim_{FS}(R) = \frac{\emph{\ent}(\alpha)}{\log k}.
\]
\end{thm}

\noindent The following corollary follows from Theorem \ref{norm}.

\begin{cor}\label{jointnorm}
If $\alpha$ is a probability measure on $\Sigma \times \Sigma$, then, for every $\alpha$-normal sequence $(R_1,R_2) \in (\Sigma \times \Sigma)^\infty$,
\[
dim_{FS}(R_1,R_2) = Dim_{FS}(R_1,R_2) = \frac{\emph{\ent}(\alpha)}{\log k}.
\]
\end{cor}

\noindent Our first theorem of this section is a ``mutual'' version of Theorem \ref{norm}.

\begin{thm}\label{normmut}
If $\alpha$ is a probability measure on $\Sigma \times \Sigma$, then, for every $\alpha$-normal sequence $(R_1,R_2) \in (\Sigma \times \Sigma)^\infty$,
\[
mdim_{FS}(R_1:R_2) = Mdim_{FS}(R_1:R_2) = \frac{\emph{\mi}(\alpha_1;\alpha_2)}{\log k}.
\]
\end{thm}

\begin{proof}
If $(R_1,R_2)$ is $\alpha$-normal, then, by Observation \ref{obs4}, $R_1$ is $\alpha_1$-normal and $R_2$ is $\alpha_2$-normal. By the properties of finite-state mutual dimension listed in Theorem \ref{md}, we have
\[
dim_{FS}(R_1) + dim_{FS}(R_2) - dim_{FS}(R_1,R_2) \leq mdim_{FS}(S:T)
\]
and
\[
mdim_{FS}(S:T) \leq Dim_{FS}(R_1) + Dim_{FS}(R_2) - dim_{FS}(R_1,R_2).
\]
Furthermore, by applying Theorem \ref{norm} and Corollary \ref{jointnorm}, we obtain
\[
\frac{\ent(\alpha_1)}{\log k} + \frac{\ent(\alpha_2)}{\log k} - \frac{\ent(\alpha)}{\log k} \leq mdim_{FS}(R_1:R_2) \leq \frac{\ent(\alpha_1)}{\log k} + \frac{\ent(\alpha_2)}{\log k} - \frac{\ent(\alpha)}{\log k}.
\]
Finally, by the definition of Shannon mutual information,
\[
\frac{\mi(\alpha_1;\alpha_2)}{\log k} \leq mdim_{FS}(R_1:R_2) \leq \frac{\mi(\alpha_1;\alpha_2)}{\log k}.
\]
A similar argument can be given to show that $Mdim_{FS}(R_1:R_2) = \frac{\mi(\alpha_1;\alpha_2)}{\log k}$. \qedhere
\end{proof}

\begin{defn}
Let $\alpha$ and $\beta$ be probability measures on $\Sigma$. The \emph{Kullback-Leibler divergence} between $\alpha$ and $\beta$ is
\[
\emph{D}(\alpha\,||\,\beta) = \displaystyle\sum_{a \in \Sigma}\alpha(a)\frac{\alpha(a)}{\beta(a)}.
\]
\end{defn}
\noindent Lutz also proved the following lemma regarding the Shannon self-information of a sequence with respect to a probability measure.
\begin{lem}[Frequency Divergence Lemma \cite{jLutz11}]\label{fdl}
If $\alpha$ and $\beta$ are positive probability measure on $\Sigma$, then, for all $S \in FREQ^{\alpha}$,
\[
\ell_{\beta}(S \upharpoonright n) = (\textnormal{\textrm{H}}(\alpha) + D(\alpha\,||\,\beta))n + o(n).
\]
\end{lem}

\begin{lem}\label{product}
If $\alpha_1$ and $\alpha_2$ are positive probability measures on $\Sigma$, then, for all $S \in FREQ^{\alpha_1}$ and $T \in FREQ^{\alpha_2}$,
\[
\ell_{\alpha_1 \times \alpha_2}((S,T) \upharpoonright n) = (\textnormal{\textrm{H}}(\alpha_1) + \textnormal{\textrm{H}}(\alpha_2))n + o(n).
\]
\end{lem}

\begin{proof}
Assume the hypothesis, then,
\begin{align*}
&\ell_{\alpha_1 \times \alpha_2}((S,T) \upharpoonright n)\\
&= \displaystyle\sum^{n-1}_{i=0}\log\frac{1}{(\alpha_1 \times \alpha_2)((S,T)[i])}\\
&= \displaystyle\sum^{n-1}_{i=0}-\log(\alpha_1 \times \alpha_2)((S,T)[i])\\
&= \displaystyle\sum^{n-1}_{i=0}-\log(\alpha_1(S[i])\alpha_2(T[i]))\\
&= \displaystyle\sum^{n-1}_{i=0}-\big[\log \alpha_1(S[i]) + \log \alpha_2(T[i]) \big]\\
&= \displaystyle\sum^{n-1}_{i=0}-\log \alpha_1(S[i]) + \displaystyle\sum^{n-1}_{i=0}-\log \alpha_2(T[i])\\
&= \displaystyle\sum^{n-1}_{i=0}\frac{1}{\log \alpha_1(S[i])} + \displaystyle\sum^{n-1}_{i=0}\frac{1}{\log \alpha_2(T[i])}\\
&= \ell_{\alpha_1}(S \upharpoonright n) + \ell_{\alpha_2}(T \upharpoonright n)
\end{align*}
By the above equality and the Frequency Divergence Lemma, we have
\begin{align*}
&\ell_{\alpha_1 \times \alpha_2}((S,T) \upharpoonright n)\\
&=(\textnormal{\textrm{H}}(\alpha_1) + \textnormal{\textrm{H}}(\alpha_2) + D(\alpha_1\,||\,\alpha_1) +  D(\alpha_2\,||\,\alpha_2))n + o(n)\\
&= (\textnormal{\textrm{H}}(\alpha_1) + \textnormal{\textrm{H}}(\alpha_2))n + o(n). \qedhere
\end{align*}
\end{proof}

\begin{lem}\label{normmdim0}
Let $\alpha_1$ and $\alpha_2$ be positive probability measures on $\Sigma$. If $R_1$ is $\alpha_1$-normal, $R_2$ is $\alpha_2$-normal, and
\[
Mdim_{FS}(R_1:R_2) = 0,
\]
then
\[
dim_{FS}(R_1,R_2) = \frac{\emph{\ent}(\alpha_1) + \emph{\ent}(\alpha_1)}{\log k}.
\]
\end{lem}

\begin{proof}
By Theorem \ref{md},
\[
dim_{FS}(R_1) + dim_{FS}(R_2) - dim_{FS}(R_1,R_2) \leq 0 \leq Dim_{FS}(R_1) + Dim_{FS}(R_2) - dim_{FS}(R_1,R_2).
\]
Therefore, by Theorem \ref{norm},
\begin{align*}
dim_{FS}(R_1,R_2) &\geq dim_{FS}(R_1) + dim_{FS}(R_2)\\
									&= \frac{\ent(\alpha_1) + \ent(\alpha_2)}{\log k}
\end{align*}
and
\begin{align*}
dim_{FS}(R_1,R_2) &\leq Dim_{FS}(R_1) + Dim_{FS}(R_2)\\
									&= \frac{\ent(\alpha_1) + \ent(\alpha_2)}{\log k}.
\end{align*}
Thus,
\[
dim_{FS}(R_1,R_2) = \frac{\ent(\alpha_1) + \ent(\alpha_2)}{\log k}. \qedhere
\]
\end{proof}

We now prove the main theorem of this section.

\begin{proof}[Proof of Theorem \ref{independent}]
Assume the hypothesis and that $(R_1,R_2)$ is $(\alpha_1 \times \alpha_2)$-normal. By Theorem \ref{normmut}, we have
\begin{align*}
Mdim_{FS}(R_1:R_2) &= \frac{\mi(\alpha_1;\alpha_2)}{\log k}\\
									 &= \frac{\ent(\alpha_1)}{\log k} + \frac{\ent(\alpha_2)}{\log k} - \frac{\ent(\alpha_1 \times \alpha_2)}{\log k}\\
									 &= 0.
\end{align*}
Now, we prove the converse. Assume that
\[
Mdim_{FS}(R_1:R_2) = 0.
\]
By Lemma \ref{normmdim0}, we have
\[
dim_{FS}(R_1,R_2) = \frac{\ent(\alpha_1) + \ent(\alpha_2)}{\log k}.
\]
By the above inequality and since $\rho_r(R_1,R_2)$ decreases in $r$, we know that, for all $r \in \mathbb{Z}^+$,
\begin{align*}
\rho_r(R_1,R_2) &> dim_{FS}(R_1,R_2)\\
						&= \frac{\ent(\alpha_1) + \ent(\alpha_2)}{\log k},
\end{align*}
which implies that
\begin{align}\label{ctoh}
C_r((R_1,R_2) \harp n) &\geq \frac{(\ent(\alpha_1) + \ent(\alpha_2))n\log k}{\log k} \nonumber\\
									 &= (\ent(\alpha_1) + \ent(\alpha_2))n,
\end{align}
for sufficiently large $n \in \mathbb{N}$. By Lemma \ref{product} and (\ref{ctoh}),
\begin{align*}
dim^{\alpha_1 \times \alpha_2}_{FS}((R_1,R_2)) &= \displaystyle\lim_{r \rightarrow \infty}\displaystyle\liminf_{n \rightarrow \infty}\rho^{\alpha_1 \times \alpha_2}_r((R_1,R_2))\\
																					 &= \displaystyle\lim_{r \rightarrow \infty}\displaystyle\liminf_{n \rightarrow \infty}\frac{C_r((R_1,R_2)\harp n)}{\ell_{\alpha_1 \times \alpha_2}((R_1,R_2) \harp n)}\\
																					 &\geq \displaystyle\lim_{r \rightarrow \infty}\displaystyle\liminf_{n \rightarrow \infty}\frac{(\ent(\alpha_1) + \ent(\alpha_2))n}{(\ent(\alpha_1) + \ent(\alpha_2))n + o(n)}\\
																					 &=\displaystyle\lim_{r \rightarrow \infty}\displaystyle\liminf_{n \rightarrow \infty}\frac{\ent(\alpha_1) + \ent(\alpha_2)}{\ent(\alpha_1) + \ent(\alpha_2) + o(1)}\\
																					 &= 1.
\end{align*}
Since the finite-state dimension of a sequence cannot exceed one, then, by Theorem \ref{normalchar}, we know that $(R_1,R_2)$ is $(\alpha_1 \times \alpha_2)$-normal. \qedhere
\end{proof}


\bibliographystyle{plain}
\bibliography{Master}

\end{document}